\documentclass[a4paper,reqno]{amsart}
\usepackage{amsmath,amssymb,mathrsfs}
\usepackage{txfonts,bm,pifont}
\usepackage{cite}
\usepackage[dvipdfm,bookmarks=true,bookmarksnumbered=true,colorlinks=true]{hyperref}
\usepackage{comment}
\usepackage[dvipdfmx]{graphicx}
\PassOptionsToPackage{svgnames}{xcolor}
\usepackage{tikz}
\usetikzlibrary{calc}
\usetikzlibrary{decorations.markings,decorations.pathmorphing}

\newtheorem{theorem}{Theorem}[section]

\newtheorem{lemma}[theorem]{Lemma}
\newtheorem{remark}[theorem]{Remark}
\newtheorem{definition}[theorem]{Definition}

\numberwithin{equation}{section}
\newcommand{\wutilde}[1]{\vrule depth 0pt width 0pt%
{\raise0.8pt\hbox{$\smash{{\mathop{#1} \limits_{\displaystyle\widetilde{}}}}$}}}

\newcommand{\hT}{\hat{T}}
\newcommand{\hR}{\hat{R}}
\newcommand{\hrho}{\hat{\rho}}
\newcommand{\al}{\alpha}
\newcommand{\be}{\beta}
\newcommand{\de}{\delta}
\newcommand{\ga}{\gamma}
\newcommand{\la}{\lambda}
\newcommand{\ep}{\bm{\epsilon}}

\newcommand{\PDE}{P$\Delta$E}
\newcommand{\bbZ}{\mathbb{Z}}
\newcommand{\bbR}{\mathbb{R}}
\newcommand{\bbC}{\mathbb{C}}
\newcommand{\noin}{\in \hspace{-0.73em}/\,}
\newcommand{\calR}{{\mathcal R}}
\newcommand{\calr}{{\mathcal r}}
\newcommand{\ii}{{\rm i}}
\newcommand{\ee}{{\rm e}}
\newcommand{\bml}{{\bm l}}
\newcommand{\bmv}{{\bm v}}
\newcommand{\spT}{\hat{T}_{\rm SP}}
\newcommand{\Set}[2]{\left\{#1\right\}_{#2}}
\newcommand{\set}[2]{\left\{\left. #1 ~\right|~ #2 \right\}}
\makeatletter
\long\def\@makecaption#1#2{
 \vskip 10pt
 \setbox\@tempboxa\hbox{#1. #2}
 \ifdim \wd\@tempboxa >\hsize #1. #2\par \else \hbox
to\hsize{\hfil\box\@tempboxa\hfil}
 \fi}
\makeatother
\begin{document}
\title[Lattice equations arising from discrete Painlev\'e systems. II]
{Lattice equations arising from discrete Painlev\'e systems. II. $A_4^{(1)}$ case}
\author{Nalini Joshi}
\author{Nobutaka Nakazono}
\address{School of Mathematics and Statistics, The University of Sydney, New South Wales 2006, Australia.}
\email{nobua.n1222@gmail.com}
\author{Yang Shi}
\begin{abstract}
In this paper, we construct two lattices from the $\tau$ functions of $A_4^{(1)}$-surface $q$-Painlev\'e equations,
on which quad-equations of ABS type appear.
Moreover, using the reduced hypercube structure, we obtain the Lax pairs of the $A_4^{(1)}$-surface $q$-Painlev\'e equations.
\end{abstract}

\maketitle
\setcounter{tocdepth}{1}
\tableofcontents
\section{Introduction}

Two longstanding classifications of integrable discrete systems in different dimensions, one by Adler-Bobenko-Suris (ABS) \cite{ABS2003:MR1962121,ABS2009:MR2503862,BollR2011:MR2846098,BollR2012:MR3010833,BollR:thesis} 
and the other by Sakai \cite{SakaiH2001:MR1882403}, 
have been widely studied, but the mathematical connection between them remains incomplete. 
How to reduce the ABS partial difference equations to Sakai's discrete Painlev\'e equations is a natural question that has inspired many authors \cite{NP1991:MR1098879,GRSWC2005:MR2117991,FJN2008:MR2425981,HHJN2007:MR2303490,OrmerodCM2012:MR2997166,HHNS2015:MR3317164,OrmerodCM2014:MR3210633}. 
However, these earlier approaches focused on taking periodic constraints in two dimensions that lead to equations with a restricted set of parameters, manually extending these by adding gauge transformations in order to introduce more parameters. 
Another rich vein of inquiry reduces the Lax pairs of ABS equations to provide these elusive linear problems for discrete Painlev\'e equations. 
We provide a different approach grounded in higher-dimensional geometry associated naturally with full-parameter discrete Painlev\'e equations\cite{JNS2015:MR3403054,JNS2014:MR3291391,JNS:paper3}. 
In this paper, we review our approach and illustrate it for $A_4^{(1)}$-surface type $q$-discrete Painlev\'e equations, providing new Lax pairs for these equations. 

The geometric setting of reflection groups is essential to our approach. 
Within this framework, we construct higher dimensional lattices, called $\omega$-lattices, from discrete Painlev\'e equations. 
These lattices also arise from integer lattices associated with ABS classification and thereby provide a bridge between the two classifications. 
In an earlier series of papers \cite{JNS2015:MR3403054,JNS2014:MR3291391,JNS:paper3}, we constructed $\omega$-lattices for $A_5^{(1)}$- and $A_6^{(1)}$- surface $q$-Painlev\'e equations.
The $A_4^{(1)}$-case is a simpler (less degenerate) surface than these earlier cases, but it is well known that when the surface is simpler, the corresponding symmetry groups and discrete Painlev\'e equations become more complex\cite{SakaiH2001:MR1882403}. 

Despite the increasing complexity, our approach connects discrete Painlev\'e equations to partial difference equations through reductions of hypercubes and polytopes. 
We construct two lattices in two ways, one through reduction of polytopes and the other by reduction of hypercubes. 
Both lattices arise from the $\tau$ functions of $A_4^{(1)}$-surface type $q$-discrete Painlev\'e equation.
They share fundamental variables (called $\omega$-variables) and both give rise to ABS equations and to $q$-discrete Painlev\'e equations. 
The polytope case will be investigated further in future work. 
The hypercube lattice is referred to below as $\omega_{A_2+A_1}$.
(More details are given in \S \ref{subsection:main_result} and \S \ref{section:affineWeylA2}.) 

A fundamental property of the ABS equations is their consistency around each cube of the integer lattice. 
The reduced hypercube structure of the lattice $\omega_{A_2+A_1}$ then provides us with reductions of the Lax pairs of ABS squations, which turn out to be new Lax pairs for $q$-Painlev\'e equations \eqref{eqns:intro_qPs}. 
Our results show that these equations share one monodromy problem. 
Moreover, the coefficient matrices in each case are factorized into product of matrices that are linear in the monodromy variable $x$. 
We remark that in each case, we also obtain Lax pairs for the scalar form of the equations.  

In this paper, we construct two important lattices, where quad-equations are observed, 
from the $\tau$ functions of $A_4^{(1)}$-surface type $q$-discrete Painlev\'e equation.
One is the $\omega$-lattice of type $A_4^{(1)}$ investigated in \S \ref{section:omega_lattice_A4}.
An $\omega$-lattice provides informations about how a system of partial difference equations can be reduced to discrete Painlev\'e equations. 
It provides not only the type of equation, but also the combinatorial structure of the lattice before reduction 
(see \cite{JNS2015:MR3403054,JNS2014:MR3291391} for details).
The other lattice is the $\omega_{A_2+A_1}$ investigated in \S \ref{section:affineWeylA2}.
The lattice $\omega_{A_2+A_1}$ can be obtained from an integer lattice, 
given by the space-filling of the hypercube on whose faces quad-equations of ABS type are assigned, by the geometric reduction.
By using this reduced hypercube structure, we obtain the Lax pairs of the $q$-Painlev\'e equations \eqref{eqns:intro_qPs}.
These Lax pairs differ from the ones in the literature \cite{MurataM2009:MR2485835}.
Moreover, our result show that four equations of $A_4^{(1)}$-surface type share 
the same $q$-discrete monodromy problem \eqref{eqn:intro_linear} 
with differing deformation equations given by \eqref{eqns:intro_deformations}.
Other distinctive properties of our Lax pairs are that 
their coefficient matrices occur as products of matrices of degree one in the spectral parameter $x$
and elements of the coefficient matrices given by the rational functions of Painlev\'e variables.

\subsection{$A_4^{(1)}$-surface $q$-Painlev\'e equations}
In this paper, we collectively call the following $q$-difference equations as $A_4^{(1)}$-surface $q$-Painlev\'e equations
since they are of $A_4^{(1)}$-surface type in Sakai's classification\cite{SakaiH2001:MR1882403}:
{\allowdisplaybreaks
\begin{subequations}\label{eqns:intro_qPs}
\begin{align}
 \text{$q$-P$_{\rm V}$:}&&&
 \begin{cases}\label{eqn:intro_qP5_Sakai_1}
 \overline{F}F=\cfrac{1}{{c_3}^2t^2}\,\cfrac{(c_1+tG)(c_2+tG)}{c_3+G},\\[1em]
 G\underline{G}=\cfrac{{c_3}^2}{t^2}\,\cfrac{({c_3}^{-1}+tF)(qc_1c_2{c_3}^{-2}+tF)}{{c_3}^{-1}+F},
 \end{cases}\\
 \text{$q$-${{\rm P}_{\rm V}}^\ast$:}&&&
 \begin{cases}\label{eqn:intro_qP5_Tamizhmani_1}
 (\overline{F}G-1)(F G-1)=\cfrac{t^2}{qc_1{c_2}^2}\,\cfrac{({c_1}^{-1}{c_3}^2+G)(qc_1{c_2}^2+G)}{c_3t+G},\\[1em]
 (F G-1) (F \underline{G}-1)=\cfrac{{c_3}^2 t^2}{qc_1}\,\cfrac{(c_1{c_3}^{-2}+F)(q^{-1} {c_1}^{-1}{c_2}^{-2}+F)}{q^{-1}{c_2}^{-1}t+F},
 \end{cases}\\
 \text{$q$-P$_{\rm III}(D_7^{(1)})$:}&&&
 \widetilde{G}\wutilde{G}=\cfrac{1}{t^2}\,\cfrac{(c_1+tG)(p^{-1}+tG)}{1+G},
 \label{eqn:intro_qP3_1}\\
 \text{$q$-P$_{\rm IV}$:}&&&
 (\widetilde{G}G-1)(G\wutilde{G} -1)=\cfrac{t^2}{p^2c_1{c_2}^2}\,\cfrac{(p^{-2}{c_1}^{-1}{c_2}^{-2}+G)(p^2c_1{c_2}^2+G)}{p^{-1}{c_2}^{-1}t+G},
 \label{eqn:intro_qP4_1}
\end{align}
\end{subequations}
where $t,c_1,c_2,c_3,q,p\in\bbC^\ast$ and
\begin{equation}
 F=F(t),\quad
 G=G(t),\quad
 \overline{F}=F(q t),\quad
 \underline{G}=G(q^{-1}t),\quad
 \widetilde{G}=G(p t),\quad 
 \wutilde{G}=G(p^{-1}t).
\end{equation}
We note that 
$q$-P$_{\rm V}$ \eqref{eqn:intro_qP5_Sakai_1}, 
$q$-${{\rm P}_{\rm V}}^\ast$ \eqref{eqn:intro_qP5_Tamizhmani_1}, 
$q$-P$_{\rm III}({D_7^{(1)}})$ \eqref{eqn:intro_qP3_1} and
$q$-P$_{\rm IV}$ \eqref{eqn:intro_qP4_1}
are known as 
a $q$-discrete analogue of the Painlev\'e V equation \cite{SakaiH2001:MR1882403},
that of the Painlev\'e V equation\cite{TGCR2004:MR2058894},
that of the Painlev\'e III equation of $D_7^{(1)}$-surface type \cite{NN2013:MR3157150} and
that of the Painlev\'e IV equation\cite{RG1996:MR1399286},
respectively.
}

\begin{remark}
It is known that
$q$-{\rm P}$_{\rm V}$ \eqref{eqn:intro_qP5_Sakai_1} and $q$-${{\rm P}_{\rm V}}^\ast$ \eqref{eqn:intro_qP5_Tamizhmani_1}
can be reduced to 
$q$-{\rm P}$_{\rm III}({D_7^{(1)}})$ \eqref{eqn:intro_qP3_1} {\rm\cite{NN2013:MR3157150}} and
$q$-{\rm P}$_{\rm IV}$ \eqref{eqn:intro_qP4_1} {\rm \cite{NakazonoN2014:MR3261854}}
by the projective reductions:
\begin{subequations}
\begin{align}
 &c_2=p^{-1},\quad c_3=1,\quad q=p^2,\quad F=\wutilde{G},\\
 &c_2c_3=p^{-1},\quad q=p^2,\quad F=\wutilde{G},
\end{align}
\end{subequations}
respectively.
In this sense, $q$-{\rm P}$_{\rm III}({D_7^{(1)}})$ and $q$-{\rm P}$_{\rm IV}$ are sometimes called as the scalar forms of 
$q$-{\rm P}$_{\rm V}$ and $q$-${{\rm P}_{\rm V}}^\ast$, respectively.
\end{remark}

\subsection{Main results}
\label{subsection:main_result}
In this section, we outline two main results of this paper.

Firstly, in \S \ref{subsection:geometric_reduction}, we prove the following theorem.
\begin{theorem}\label{maintheorem_hypercubestructure}
The lattice $\omega_{A_2+A_1}$ has a reduced hypercube structure.
\end{theorem}
The lattice $\omega_{A_2+A_1}$ is a 3-dimensional integer lattice on which ABS equations \eqref{eqns:A4_H3}--\eqref{eqns:A4_D4} 
and $q$-Painlev\'e equations \eqref{eqns:intro_qPs} appear.
This lattice is constructed from the $\tau$ functions of $A_4^{(1)}$-surface $q$-Painlev\'e equations (see \S \ref{section:affineWeylA2}).
Theorem \ref{maintheorem_hypercubestructure} means that 
the lattice $\omega_{A_2+A_1}$ can be also obtained from the 4-dimensional hypercube lattice on whose faces ABS equations are assigned.
This reduced hypercube structure turn out to be essential in the construction of Lax pairs for discrete Painlev\'e equations\cite{JNS:paper3}.

Our second main result, Theorem \ref{maintheorem_Lax}, concerns the Lax pairs of the $q$-Painlev\'e equations \eqref{eqns:intro_qPs}.
Equations \eqref{eqns:intro_qPs} share one spectral linear problem, which takes the factorized form
\begin{equation}\label{eqn:intro_linear}
 \Phi(px)
 =\begin{pmatrix}\ast \,x&\ast\\\ast&0\end{pmatrix}.
 \begin{pmatrix}\ast \,x&\ast\\\ast&\ast \,x\end{pmatrix}.
 \begin{pmatrix}\ast \,x&\ast\\\ast&\ast \,x\end{pmatrix}.
 \begin{pmatrix}\ast \,x&\ast\\\ast&\ast \,x\end{pmatrix}.\Phi(x)
 =A.\Phi(x).
\end{equation}
Here, the $2\times2$ matrix $A=A(x)$ is given by \eqref{eqn:intro_linear_martrix}
whose elements are expressed by the non-zero complex parameters $b_i$, $i=0,\dots,3$, and $p$ 
and unknown functions $f_i^{(1)}$, $i=1,2,3$.
Note that the functions $f_i^{(1)}$ satisfy the following relation$:$
\begin{equation}\label{eqn:relation_f_A2A1}
 b_1 b_2 b_3
 +p {b_1}^3 b_2 f_1^{(1)}
 -p^2{b_0}^3 {b_3}^{3/2} f_1^{(3)}
 +p^2b_0 {b_1}^4 {b_3}^{3/2} f_1^{(1)} f_1^{(2)} f_1^{(3)}
 =0.
\end{equation}
We introduce the deformation operators $T_0$, $T_{13}$, $R_0$ and $R_{13}$
whose actions on the parameters $b_i$, $i=0,\dots,3$, and $p$ are given by
\begin{subequations}\label{eqns:intro_action_para}
\begin{align}
 T_0&:(b_0,b_1,b_2,b_3,p)\mapsto (p b_0,p b_1,b_2,b_3,p),\\
 T_{13}&:(b_0,b_1,b_2,b_3,p)\mapsto (b_0,b_1,p^2b_2,b_3,p),\\
 R_0&:(b_0,b_1,b_2,b_3,p)\mapsto (b_1,pb_0,b_2,{b_3}^{-1},p),\\
 R_{13}&:(b_0,b_1,b_2,b_3,p)\mapsto (b_0,b_1,p b_2,{b_3}^{-1},p),
\end{align}
\end{subequations}
while those on the spectral parameter $x$ and the wave function $\Phi=\Phi(x)$ are given by
\begin{subequations}\label{eqns:intro_deformations}
\begin{align}
 &T_0(x)=T_{13}(x)=R_0(x)=R_{13}(x)=x,\\
 &T_0(\Phi)
 =\begin{pmatrix}\ast \,x&\ast\\\ast&\ast \,x\end{pmatrix}.\begin{pmatrix}\ast \,x&\ast\\\ast&\ast \,x\end{pmatrix}.\Phi(x)
 =B_{T0}.\Phi,\\
 &T_{13}(\Phi)
 =\begin{pmatrix}\ast \,x&\ast\\\ast&0\end{pmatrix}.\begin{pmatrix}\ast \,x&\ast\\\ast&0\end{pmatrix}.\Phi(x)
 =B_{T13}.\Phi,\\
 &R_0(\Phi)
 =\begin{pmatrix}\ast \,x&\ast\\\ast&\ast \,x\end{pmatrix}.\Phi(x)
 =B_{R0}.\Phi,\\
 &R_{13}(\Phi)
 =\begin{pmatrix}\ast \,x&\ast\\\ast&0\end{pmatrix}.\Phi(x)
 =B_{R13}.\Phi,
\end{align}
\end{subequations}
where the $2\times2$ matrices $B_{T0}=B_{T0}(x)$, $B_{T13}=B_{T13}(x)$, $B_{R0}=B_{R0}(x)$ and $B_{R13}=B_{R13}(x)$ 
are given by \eqref{eqns:intro_matrices_B}.
Equations \eqref{eqns:intro_action_para} and \eqref{eqns:intro_deformations} provide us with the deformation of the spectral problem.
\begin{theorem}\label{maintheorem_Lax}
The compatibility conditions of the linear equation \eqref{eqn:intro_linear} with the operators 
$T_0$, $T_{13}$, $R_0$ and $R_{13}$$:$
\begin{subequations}
\begin{align}
 &T_0(A).B_{T0}=B_{T0}(px).A,\quad
 T_{13}(A).B_{T13}=B_{T13}(px).A,\\
 &R_0(A).B_{R0}=B_{R0}(px).A,\quad
 R_{13}(A).B_{R13}=B_{R13}(px).A,
\end{align}
\end{subequations}
are equivalent to
{\allowdisplaybreaks
\begin{subequations}\label{eqn:intro_action_var}
\begin{align}
&\begin{cases}
 T_0(f_1^{(3)})f_1^{(3)}
 =\cfrac{b_1(b_3+p{b_1}^2 f_1^{(1)}) (-b_0 b_2 {b_3}^{1/2}+p{b_1}^3 f_1^{(1)})}{p^4{b_0}^4 {b_3}^2(p {b_0}^2 b_3+{b_1}^2 f_1^{(1)})},\\
 {T_0}^{-1}(f_1^{(1)})f_1^{(1)}
 =\cfrac{p b_0 ({b_3}^{-1}+p{b_0}^2 f_1^{(3)}) (-p^{-1}b_1 b_2 {b_3}^{-1/2}+p{b_0}^3 f_1^{(3)})}{{b_1}^4{b_3}^{-2}({b_1}^2{b_3}^{-1}+p{b_0}^2 f_1^{(3)})},
\end{cases}\label{eqn:intro_qP5_Sakai_2}\\
&\begin{cases}
 \left(T_{13}(f_1^{(1)})f_1^{(2)}-\cfrac{{b_0}^2}{{b_1}^4}\right)\left(f_1^{(1)} f_1^{(2)}-\cfrac{{b_0}^2}{{b_1}^4}\right)
 =\cfrac{{b_0}^2 {b_2}^2(p{b_0}^2+{b_1}^2 b_3 f_1^{(2)})(1+ p{b_1}^2b_3 f_1^{(2)})}{p{b_1}^7{b_3}^{1/2}(-b_0 b_2+{b_1}^3 {b_3}^{1/2} f_1^{(2)})},\\
 \left(f_1^{(1)} f_1^{(2)}-\cfrac{{b_0}^2}{{b_1}^4}\right) \left(f_1^{(1)}{T_{13}}^{-1}(f_1^{(2)})-\cfrac{{b_0}^2}{{b_1}^4}\right)
 =\cfrac{{b_0}^2 {b_2}^2(p{b_0}^2b_3+{b_1}^2 f_1^{(1)})(b_3+p{b_1}^2 f_1^{(1)})}{p^2{b_1}^7b_3(-b_0 b_2 {b_3}^{1/2}+p{b_1}^3 f_1^{(1)})},
\end{cases}\label{eqn:intro_qP5_Tamizhmani_2}\\
&\begin{cases}
 R_0(f_1^{(3)})=f_1^{(1)},\\
 R_0(f_1^{(1)})f_1^{(3)}
 =\cfrac{b_1(b_3+p{b_1}^2 f_1^{(1)}) (-b_0 b_2 {b_3}^{1/2}+p{b_1}^3 f_1^{(1)})}{p^4{b_0}^4 {b_3}^2(p {b_0}^2 b_3+{b_1}^2 f_1^{(1)})},
\end{cases}\label{eqn:intro_qP3_2}\\
&\begin{cases}
 R_{13}(f_1^{(1)})=f_1^{(2)},\\
 \left(R_{13}(f_1^{(2)})f_1^{(2)}-\cfrac{{b_0}^2}{{b_1}^4}\right)\left(f_1^{(1)} f_1^{(2)}-\cfrac{{b_0}^2}{{b_1}^4}\right)
 =\cfrac{{b_0}^2 {b_2}^2(p{b_0}^2+{b_1}^2 b_3 f_1^{(2)})(1+ p{b_1}^2b_3 f_1^{(2)})}{p{b_1}^7{b_3}^{1/2}(-b_0 b_2+{b_1}^3 {b_3}^{1/2} f_1^{(2)})},
\end{cases}\label{eqn:intro_qP4_2}
\end{align}
\end{subequations}
respectively.
}
\end{theorem}

This theorem is proven in \S \ref{subsection:Laxpairs}.
The actions \eqref{eqns:intro_action_para} and \eqref{eqn:intro_action_var} correspond to the $q$-Painlev\'e equations \eqref{eqns:intro_qPs}.
\begin{remark}\label{remark:relation_qps}
Equations \eqref{eqn:intro_qP5_Sakai_2} and \eqref{eqn:intro_qP5_Tamizhmani_2} are equivalent to 
$q$-{\rm P}$_{\rm V}$ \eqref{eqn:intro_qP5_Sakai_1} and
$q$-${{\rm P}_{\rm V}}^\ast$ \eqref{eqn:intro_qP5_Tamizhmani_1}
by the following correspondences:
\begin{subequations}
\begin{align}
 &\bar{}=T_0,\quad
 t={b_1}^2,\quad
 c_1=-\cfrac{b_0b_2{b_3}^{1/2}}{pb_1},\quad
 c_2=\cfrac{b_3}{p},\quad
 c_3=\cfrac{p{b_0}^2b_3}{{b_1}^2},\quad
 q=p^2,\notag\\
 &F=f_1^{(3)},\quad
 G=f_1^{(1)},
 \label{eqn:intro_correspondence_T0}\\[0.5em]
 &\bar{}=T_{13},\quad
 t=p^{1/2}b_2,\quad
 c_1=-\cfrac{b_0}{{b_1}^2},\quad
 c_2=\cfrac{b_1}{p^{1/2}{b_3}^{1/2}},\quad
 c_3=\cfrac{1}{p^{1/2}b_1{b_3}^{1/2}},\notag\\
 &q=p^2,\quad
 F=-\cfrac{{b_1}^2}{b_0}\,f_1^{(1)},\quad
 G=-\cfrac{{b_1}^2}{b_0}\,f_1^{(2)},
 \label{eqn:intro_correspondence_T13}
\end{align}
\end{subequations}
respectively.
Moreover, letting 
\begin{equation}
 b_1=p^{1/2}b_0,\quad b_3=1,
\end{equation}
and setting
\begin{equation}
 \tilde{}=R_0,\quad
 t={b_1}^2,\quad
 c_1=-\cfrac{b_2}{p^{3/2}},\quad
 G=f_1^{(1)},
 \label{eqn:intro_correspondence_R0}
\end{equation}
we obtain $q$-{\rm P}$_{\rm III}({D_7^{(1)}})$ \eqref{eqn:intro_qP3_1} from the action \eqref{eqn:intro_qP3_2}.
Similarly, by letting 
\begin{equation}
 b_3=1,
\end{equation}
and setting
\begin{equation}
 \tilde{}=R_{13},\quad
 t=p^{1/2}b_2,\quad
 c_1=-\cfrac{b_0}{{b_1}^2},\quad
 c_2=\cfrac{b_1}{p^{1/2}},\quad
 G=-\cfrac{{b_1}^2}{b_0}\,f_1^{(2)},
 \label{eqn:intro_correspondence_R13}
\end{equation}
the action \eqref{eqn:intro_qP4_2} gives $q$-{\rm P}$_{\rm IV}$ \eqref{eqn:intro_qP4_1}.
\end{remark}

\subsection{Background}
\label{subsection:background}
In the 1900s, in order to find new class of special functions, Painlev\'e and Gambier classified all differential equations in the form of 
$y''=F(y',y,t)$, where $y=y(t)$, $'=d/dt$ and $F$ is a rational function, by imposing the condition that the solutions should admit only poles as movable singular points. 
As a result, they showed that the resulting equations can be reduced to one  of the six equations,
which are now called the Painlev\'e I through VI equations,
unless it can be integrated algebraically, or transformed into a simpler equations 
such as a linear equation or the differential equation of elliptic functions.
Moreover, it is known that Painlev\'e equations can be classified into eight types 
by the geometrical classification of space of initial conditions\cite{OKSO2006:MR2277519,OkamotoK1979:MR614694,SakaiH2001:MR1882403}.
From the view point of this classification, P$_{\rm III}$ can be divided into 
P$_{\rm III}(D_6^{(1)})$, P$_{\rm III}(D_7^{(1)})$ and P$_{\rm III}(D_8^{(1)})$ by the values of parameters.

Discrete Painlev\'e equations are nonlinear ordinary difference equations of second order, 
which include discrete analogues of the Painlev\'e equations. 
The geometric classification of discrete Painlev\'e equations,
based on types of rational surfaces connected to affine Weyl groups,
is well known\cite{SakaiH2001:MR1882403}.
Together with the Painlev\'e equations, 
they are now regarded as one of the most important classes of equations 
in the theory of integrable systems (see, e.g., \cite{GR2004:MR2087743,KNY2015:arXiv150908186K}). 

It is well known that the $\tau$ functions, which gives rise to various bilinear equations,
play a crucial role in the theory of integrable systems\cite{book_MJD2000:MR1736222}.
The same is true in the theory of continuous and discrete Painlev\'e equations
\cite{JMU1981:MR630674,JM1981:MR625446,JM1981:MR636469,book_NoumiM2004:MR2044201,OkamotoK1987:MR916698,OkamotoK1987:MR914314,OkamotoK1986:MR854008,OkamotoK1987:MR927186}.
A representation of the affine Weyl groups can be lifted to the level of the $\tau$ functions
\cite{TsudaT2006:MR2207047,KMNOY2003:MR1984002,KMNOY2006:MR2353465,JNS2015:MR3403054,TM2006:MR2202304,MasudaT2011:MR2765599,MasudaT2009:MR2506177}.

Discrete Painlev\'e equations are called integrable because they arise as compatibility conditions of associated linear problems called Lax pairs. 
The search for and construction of Lax pairs of discrete Painlev\'e equations has been a very active research area.
Noteworthy approaches include extensions of Birkhoff's study of linear $q$-difference equations\cite{JS1996:MR1403067,SakaiH2006:MR2266221,SakaiH2005:MR2177121},
periodic-type reductions from ABS equations or the discrete KP/UC hierarchy\cite{HayM2007:MR2371129,HHJN2007:MR2303490,OVHQ2014:MR3215839,OrmerodCM2012:MR2997166,OVQ2013:MR3030178,KNY2002:MR1958118,TsudaT2010:MR2563787,PNGR1992:MR1162062,JNS:paper3},
extensions of Schlesinger transformations
\cite{DT2014:arXiv1408.3778,DST2014:arXiv1302.2972,BoalchP2009:MR2500553},
search for linearizable curves in the space of initial values \cite{YamadaY2009:MR2506170,YamadaY2011:MR2836394},
Pad\'e approximation or interpolation\cite{IkawaY2013:MR3061504,NagaoH2015:MR3323664,NTY2013:MR3077695} 
and the theory of orthogonal polynomials\cite{WO2012:MR3007262,WitteNS2015:MR3413576,OWF2011:MR2819929,BianeP2014:MR3221944,BA2010:MR2578525,BB2003:MR1962463}.

\subsection{Plan of the paper}
This paper is organized as follows:
in \S \ref{section:omega_lattice_A4}, we introduce the $\tau$ functions of $A_4^{(1)}$-surface $q$-Painlev\'e equations, which have the extended affine Weyl group symmetry $\widetilde{W}(A_4^{(1)})$,
and show that the $q$-Painlev\'e equations \eqref{eqns:intro_qPs} can be derived from a birational representation of $\widetilde{W}(A_4^{(1)})$.
Moreover, we construct the $\omega$-lattice of type $A_4^{(1)}$ 
and then derive various quad-equations of ABS type, as relations on the $\omega$-lattice.
In \S \ref{section:affineWeylA2}, we construct the lattice $\omega_{A_2+A_1}$ and show its properties.
In \S \ref{section:proof_theorem}, we give the proofs of Theorems \ref{maintheorem_hypercubestructure} and \ref{maintheorem_Lax}
by using the geometric reduction from the integer lattice $\bbZ^4$ with the integrable \PDE s to the lattice $\omega_{A_2+A_1}$.
Some concluding remarks are given in \S \ref{ConcludingRemarks}.
\section{Construction of the $\omega$-lattice of type $A_4^{(1)}$}
\label{section:omega_lattice_A4}
In this section, we define $\tau$ functions by using the transformation group $\widetilde{W}(A_4^{(1)})$.
Then, we derive the $q$-Painlev\'e equations \eqref{eqns:intro_qPs}
and construct the $\omega$-lattice of type $A_4^{(1)}$ from the $\tau$ functions.

For convenience, throughout this paper we use the following notation for compositions of arbitrary mappings $w_i$, $i=1,\dots,n$:
\begin{equation}
 w_1\cdots w_n:=w_1\circ\cdots\circ w_n.
\end{equation}

\subsection{$\tau$ functions}
In this section, we define the $\tau$ functions by using the transformation group 
$\widetilde{W}(A_4^{(1)})=\langle s_0,s_1,s_2,s_3,s_4,\sigma,\iota\rangle$,
which forms the extended affine Weyl group of type $A_4^{(1)}$ (see Appendix \ref{section:proof_tauA4}).

Below, we describe the actions of $\widetilde{W}(A_4^{(1)})$ on the five parameters $a_0,\dots,a_4\in\bbC^\ast$ 
and on the ten variables $\tau_i^{(j)}$, $i=1,2$, $j=1,\dots,5$, which satisfy the following three relations:
\begin{subequations}\label{eqn:A4_conditions_tau}
\begin{align}
 &\tau_2^{(1)}=\cfrac{a_0 a_1 (a_3 \tau_1^{(3)} \tau_1^{(5)}+a_0 \tau_1^{(4)} \tau_2^{(3)})}{a_2 {a_3}^2 \tau_2^{(5)}},\\
 &\tau_2^{(2)}=\cfrac{a_1 a_2 (a_4 \tau_1^{(1)} \tau_1^{(4)}+a_1 \tau_1^{(5)} \tau_2^{(4)})}{a_3 {a_4}^2 \tau_2^{(1)}},\\
 &\tau_2^{(4)}=\cfrac{a_3 a_4 (a_1 \tau_1^{(1)} \tau_1^{(3)}+a_3 \tau_1^{(2)} \tau_2^{(1)})}{a_0 {a_1}^2 \tau_2^{(3)}}.
\end{align}
\end{subequations}

\begin{remark}
Below we use the index $j$ to denote an element of $\bbZ/5\bbZ$ with a slightly different enumeration 
for transformations $s_0$, \dots, $s_4$, parameters $a_0$, \dots, $a_4$ and variables $\tau_i^{(1)}$, \dots, $\tau_i^{(5)}$ $(i=1,2)$.
To avoid confusion, we point out, for example, that $j=0$ for $s_j$ and $a_j$ would imply $j=5$ for $\tau_i^{(j)}$.
\end{remark}
\begin{lemma}\label{lemma:tau_A4}
The action of $\widetilde{W}(A_4^{(1)})$ on the parameters are given by
\begin{equation}\label{eqn:A4_weylaction_a}
 s_i(a_j)=a_j{a_i}^{-a_{ij}},\quad 
 \sigma(a_i)=a_{i+1},\quad
 \iota(a_i)={a_{-i}}^{-1},
\end{equation}
where $i,j\in\bbZ/5\bbZ$ and
\begin{equation}
 (a_{ij})_{i,j=0}^4
 =\left(\begin{array}{ccccc}
 2&-1&0&0&-1\\-1&2&-1&0&0\\0&-1&2&-1&0\\0&0&-1&2&-1\\-1&0&0&-1&2\end{array}\right)
\end{equation}
is the Cartan matrix of type $A_4^{(1)}$,
while their actions on the variables are given by
\begin{subequations}\label{eqns:A4_weylaction_tau}
\begin{align}
 &s_j(\tau_1^{(j)})=\tau_2^{(j+4)},\quad
 s_j(\tau_2^{(j+3)})
 =\cfrac{a_{j+3} a_{j+4}(a_j a_{j+1} \tau_1^{(j+1)} \tau_1^{(j+3)}+a_{j+3} \tau_1^{(j+2)} \tau_2^{(j+1)})}
  {{a_{j+1}}^2 \tau_1^{(j)}},\\
 &s_j(\tau_2^{(j+4)})=\tau_1^{(j)},\quad
 s_j(\tau_2^{(j)})
 =\cfrac{a_{j+4}(a_{j+2} \tau_1^{(j+2)} \tau_1^{(j+4)}+a_j a_{j+4} \tau_1^{(j+3)} \tau_2^{(j+2)})}
  {a_j a_{j+1} {a_{j+2}}^2 \tau_1^{(j)}},\\
 &\sigma(\tau_1^{(j)})=\tau_1^{(j+1)},\quad
 \sigma(\tau_2^{(j)})=\tau_2^{(j+1)},\quad
 \iota(\tau_1^{(j)})=\tau_1^{(5-j)},\quad
 \iota(\tau_2^{(j)})=\tau_2^{(3-j)},
\end{align}
\end{subequations}
where $j\in\bbZ/5\bbZ$.
In general, for a function $F=F(a_i,\tau_j^{(k)})$, we let an element
$w\in\widetilde{W}(A_4^{(1)})$ act as $w.F=F(w.a_i,w.\tau_j^{(k)})$, that is, 
$w$ acts on the arguments from the left. 
\end{lemma}

The proof of Lemma \ref{lemma:tau_A4} is given in Appendix \ref{section:proof_tauA4}.

\begin{remark}
The action of $\widetilde{W}(A_4^{(1)})$ in Lemma \ref{lemma:tau_A4} was first obtained by Tsuda in \cite{TsudaT2006:MR2207047} without the details of the proof.
The notations in this paper are related to those in \cite{TsudaT2006:MR2207047}
by the following correspondence:
\begin{subequations}
\begin{align}
 &(s_0,s_1,s_2,s_3,s_4,\sigma,\iota)\to(s_0,s_1,s_2,s_3,s_4,\pi^3,\iota),\\
 &(a_0,a_1,a_2,a_3,a_4,q)\to(a_0,a_1,a_2,a_3,a_4,q),\\
 &(\tau_1^{(1)},\tau_1^{(2)},\tau_1^{(3)},\tau_1^{(4)},\tau_1^{(5)})
 \to(\tau_2,\tau_3,\tau_6,\tau_5,\tau_7),\\
 &(\tau_2^{(1)},\tau_2^{(2)},\tau_2^{(3)},\tau_2^{(4)},\tau_2^{(5)})
 \to(\pi^3(\tau_1),\pi(\tau_1),\tau_4,\pi^3(\tau_4),\tau_1).
\end{align}
\end{subequations}
We also note that in \cite{TsudaT2006:MR2207047} each element $w\in\widetilde{W}(A_4^{(1)})$
acts on the arguments from the right, whereas in the present paper it acts from the left. 
\end{remark}

To iterate each variable $\tau_i^{(j)}$, we need the following transformations:
\begin{subequations}\label{eqn:A4_translations}
\begin{align}
 &T_0=\sigma s_4 s_3 s_2 s_1,\quad
 T_1=\sigma s_0 s_4 s_3 s_2,\quad
 T_2=\sigma s_1 s_0 s_4 s_3,\\
 &T_3=\sigma s_2 s_1 s_0 s_4,\quad
 T_4=\sigma s_3 s_2 s_1 s_0,
\end{align}
\end{subequations} 
which are translations on the root lattice $\hat{Q}(A_4^{(1)})$  \eqref{eqn:root_symmetry_A4} 
(see Appendix \ref{section:proof_tauA4}).
Note that $T_i$, $i=0,\dots,4$, commute with each other and 
\begin{equation}
 T_0 T_1 T_2 T_3 T_4=1.
\end{equation}
Their actions on the parameters are given by
\begin{equation}
 T_i(a_i)=qa_i,\quad 
 T_i(a_{i+1})=q^{-1}a_{i+1},\quad
 i\in\bbZ/5\bbZ,
\end{equation}
where $q=a_0a_1a_2a_3a_4$ is invariant under the actions of $T_i$.
We define $\tau$ functions by
\begin{equation}\label{eqn:A4_tau_l_0123}
 \tau_{l_0,l_1,l_2,l_3,l_4}={T_0}^{l_0}{T_1}^{l_1}{T_2}^{l_2}{T_3}^{l_3}{T_4}^{l_4}(\tau_2^{(3)}),
\end{equation}
where $l_i\in\bbZ$.
We note that 
\begin{subequations}
\begin{align}
 &\tau_1^{(1)}=\tau_{1,0,0,1,0},\quad
 \tau_1^{(2)}=\tau_{1,1,0,1,0},\quad
 \tau_1^{(3)}=\tau_{1,1,1,1,0},\quad
 \tau_1^{(4)}=\tau_{1,1,1,2,0},\\
 &\tau_1^{(5)}=\tau_{0,0,0,1,0},\quad
 \tau_2^{(1)}=\tau_{1,0,1,1,0},\quad
 \tau_2^{(2)}=\tau_{1,1,0,2,0},\quad
 \tau_2^{(3)}=\tau_{0,0,0,0,0},\\
 &\tau_2^{(4)}=\tau_{2,1,1,2,0},\quad
 \tau_2^{(5)}=\tau_{0,1,0,1,0}.
\end{align}
\end{subequations}

\subsection{Discrete Painlev\'e equations}
\label{subsection:A4_discrete_Painleve}
In this section, we define the $f$-variables by rational functions of the $\tau$-variables.
Then, we demonstrate that elements of infinite order of $\widetilde{W}(A_4^{(1)})$ give various $q$-Painlev\'e equations.

Let us define the ten $f$-variables by
\begin{equation}\label{eqn:def_f}
 f_1^{(j)}=\cfrac{\tau_1^{(j+1)}\tau_2^{(j)}}{\tau_1^{(j)}\tau_1^{(j+2)}},\quad
 f_2^{(j)}=s_{j+2}(f_1^{(j)})=\cfrac{a_ja_{j+1}(a_{j+2}a_{j+3}+a_jf_1^{(j+3)})}{{a_{j+3}}^2f_1^{(j+1)}},
\end{equation}
where $j\in\bbZ/5\bbZ$.
From the definition above and the relations \eqref{eqn:A4_conditions_tau},
the following relations hold:
\begin{equation}\label{eqns:A4_conditions_f}
 a_{j+2}{a_{j+3}}^2f_1^{(j)}f_1^{(j+1)}=a_ja_{j+1}(a_{j+3}+a_jf_1^{(j+3)}),
\end{equation}
where $j\in\bbZ/5\bbZ$.
The relations above look like five equations, but the relations represent only three.
Therefore, there are only two essential $f$-variables.
The action of $\widetilde{W}(A_4^{(1)})$  on these variables $f_i^{(j)}$ is given by the lemma below, 
which follows from the actions \eqref{eqns:A4_weylaction_tau}.

\begin{lemma}
The action of $\widetilde{W}(A_4^{(1)})$ on variables $f_i^{(j)}$ is given by
{\allowdisplaybreaks
\begin{subequations}
\begin{align}
 &s_j(f_1^{(j+3)})=f_2^{(j+3)},\quad
 s_j(f_1^{(j)})=\cfrac{a_{j+4}(a_{j+2}+a_ja_{j+4}f_1^{(j+2)})}{a_ja_{j+1}{a_{j+2}}^2f_1^{(j+4)}},\quad
 s_j(f_2^{(j+3)})=f_1^{(j+3)},\\
 &s_j(f_2^{(j+2)})=\cfrac{a_ja_{j+3}a_{j+4}(a_{j+2}+a_ja_{j+4}f_1^{(j+2)}+a_ja_{j+1}a_{j+2}f_1^{(j+4)})}{a_{j+1}f_1^{(j+4)}f_2^{(j+3)}},\\
 &s_j(f_2^{(j+4)})=\cfrac{a_ja_{j+1}{a_{j+2}}^2f_1^{(j+4)}f_1^{(j)}f_2^{(j+4)}}{a_{j+4}(a_{j+2}+a_ja_{j+4}f_1^{(j+2)})},\\
 &s_j(f_2^{(j)})=\cfrac{a_ja_{j+1}a_{j+4}+a_{j+3}a_{j+4}f_1^{(j+1)}+a_j{a_{j+1}}^2a_{j+2}f_1^{(j+4)}}{a_ja_{j+1}a_{j+3}f_1^{(j+1)}f_1^{(j+4)}},\quad
 \sigma(f_1^{(j)})=f_1^{(j+1)},\\
 &\sigma(f_2^{(j)})=f_2^{(j+1)},\quad
 \iota(f_1^{(j)})=f_1^{(3-j)},\quad
 \iota(f_2^{(j)})=\cfrac{a_{2-j}(a_{5-j}+a_{2-j}a_{3-j}f_1^{(5-j)})}{a_{3-j}a_{4-j}{a_{5-j}}^2f_1^{(2-j)}},
\end{align}
\end{subequations}
where $j\in\bbZ/5\bbZ$.
}
\end{lemma}

It is well known that the translation part of $\widetilde{W}(A_4^{(1)})$ give discrete Painlev\'e equations \cite{SakaiH2001:MR1882403}.
For examples, from the translations $T_i$, $i=1,\dots,4$, we obtain $q$-P$_{\rm V}$ \eqref{eqn:intro_qP5_Sakai_1}
and from the translations $T_iT_j$, $0\leq i<j\leq4$, we obtain $q$-${{\rm P}_{\rm V}}^\ast$ \eqref{eqn:intro_qP5_Tamizhmani_1}.
Indeed, the action of $T_0$:
\begin{subequations}
\begin{align}
 &T_0:(a_0,a_1,a_2,a_3,a_4)\mapsto(qa_0,q^{-1}a_1,a_2,a_3,a_4),\\
 &T_0(f_1^{(3)})f_1^{(3)}
 =\cfrac{a_3}{{a_0}^2{a_1}^2a_4}\,\cfrac{(a_1+a_3a_4f_1^{(1)})(a_1+a_3f_1^{(1)})}{a_0a_1+a_3f_1^{(1)}},\\
 &{T_0}^{-1}(f_1^{(1)})f_1^{(1)}
 =\cfrac{a_0{a_1}^3}{{a_3}^2}\,\cfrac{(a_2a_3+a_0f_1^{(3)})(a_3+a_0f_1^{(3)})}{a_3+a_0a_1f_1^{(3)}},
\end{align}
\end{subequations}
leads to $q$-P$_{\rm V}$ \eqref{eqn:intro_qP5_Sakai_1} by the correspondences \eqref{eqn:intro_correspondence_T0} and
\begin{subequations}\label{eqns:para_a_b}
\begin{align}
 &b_0=\cfrac{{a_0}^{1/2}}{q^{1/2}},\quad
 b_1=\cfrac{1}{{a_1}^{1/2}{a_2}^{1/2}},\quad
 b_2=-\cfrac{q^{1/4}}{{a_2}^{1/2}{a_4}^{1/2}},\\
 &b_3=\cfrac{a_0a_1a_4}{q^{1/2}},\quad
 p=q^{1/2},
\end{align}
or, equivalently,
\begin{align}
 &a_0=p^2{b_0}^2,\quad
 a_1=-\cfrac{b_2 {b_3}^{1/2}}{pb_0 b_1},\quad
 a_2=-\cfrac{pb_0}{b_1 b_2 {b_3}^{1/2}},\quad
 a_3=-\cfrac{b_1 b_2}{b_0{b_3}^{1/2}},\\
 &a_4=-\cfrac{b_1 {b_3}^{1/2}}{b_0 b_2},\quad
 q=p^2.
\end{align}
\end{subequations}
Moreover, the action of $T_{13}=T_1T_3$:
\begin{subequations}
\begin{align}
 &T_{13}:(a_0,a_1,a_2,a_3,a_4)\mapsto(a_0,qa_1,q^{-1}a_2,qa_3,q^{-1}a_4),\\
 &\left(T_{13}(f_1^{(1)})f_1^{(2)}-\cfrac{a_1 a_2}{a_3 a_4}\right)\left(f_1^{(1)} f_1^{(2)}-\cfrac{a_1 a_2}{a_3 a_4}\right)
 =\cfrac{{a_1}^3 a_2}{a_3 {a_4}^2}\,\cfrac{(a_2+a_0 a_4 f_1^{(2)}) (a_2+a_4 f_1^{(2)})}{a_1 a_2+a_4 f_1^{(2)}},\\
 &\left(f_1^{(1)} f_1^{(2)}-\cfrac{a_1 a_2}{a_3 a_4}\right) \left(f_1^{(1)}{T_{13}}^{-1}(f_1^{(2)})-\cfrac{a_1 a_2}{a_3 a_4}\right)
 =\cfrac{a_1 a_2}{a_0 {a_3}^2 {a_4}^2}\,\cfrac{(a_1+a_3 f_1^{(1)})(a_0 a_1+a_3 f_1^{(1)})}{a_1+a_3 a_4 f_1^{(1)}},
\end{align}
\end{subequations}
gives $q$-${{\rm P}_{\rm V}}^\ast$ \eqref{eqn:intro_qP5_Tamizhmani_1} by the correspondences \eqref{eqn:intro_correspondence_T13} and \eqref{eqns:para_a_b}.

It is also known that discrete Painlev\'e equations can be obtained 
from elements of infinite order of $\widetilde{W}(A_4^{(1)})$ 
which are not necessarily translations of $\widetilde{W}(A_4^{(1)})$\cite{TakenawaT2003:MR1996297,KNT2011:MR2773334}.
We here show that how $q$-P$_{\rm III}({D_7^{(1)}})$ \eqref{eqn:intro_qP3_1} and $q$-P$_{\rm IV}$ \eqref{eqn:intro_qP4_1}
can be derived from the actions of $\widetilde{W}(A_4^{(1)})$.
Let 
\begin{equation}
 R_0=\sigma^3s_2s_1,\quad
 R_{13}=\sigma s_0s_2s_4,
\end{equation}
where ${R_0}^2=T_0$ and ${R_{13}}^2=T_{13}$.
Actions of these transformations in the parameter space are not translational motion:
\begin{subequations}
\begin{align}
 R_0&:(a_0,a_1,a_2,a_3,a_4)\mapsto(a_0 a_3 a_4,q^{-1}a_1a_2a_3,a_4,a_0 a_1,a_2),\\
 R_{13}&:(a_0,a_1,a_2,a_3,a_4)\mapsto(a_0,a_1 a_2 a_3,{a_3}^{-1},q{a_2}^{-1},q^{-1}a_2a_3a_4),
\end{align}
\end{subequations}
but under the special values of the parameters these actions become translational motion.
Indeed, by imposing 
\begin{equation}\label{eqn:condition_a_qP3D7}
 a_2=a_4,\quad
 a_3=a_0a_1,
\end{equation}
which implies
\begin{equation}
 q^{1/2}=a_2a_3=a_0a_1a_2=a_3a_4=a_0a_1a_4,
\end{equation}
the action of $R_0$ becomes 
\begin{equation}
 R_0:(a_0,a_1,a_2,a_3,a_4)\mapsto(q^{1/2}a_0,q^{-1/2}a_1,a_2,a_3,a_4).
\end{equation}
Similarly, under the condition of the parameters
\begin{equation}\label{eqn:condition_a_qP4}
 q^{1/2}=a_2a_3=a_0a_1a_4,
\end{equation}
the action of $R_{13}$ becomes
\begin{equation}
 R_{13}:(a_0,a_1,a_2,a_3,a_4)\mapsto(a_0,q^{1/2} a_1,q^{-1/2}a_2,q^{1/2}a_3,q^{-1/2}a_4).
\end{equation}
Therefore, the action of $R_0$:
\begin{equation}
 R_0(f_1^{(3)})=f_1^{(1)},\quad
 R_0(f_1^{(1)})f_1^{(3)}=\cfrac{a_3}{{a_0}^2{a_1}^2a_4}\,\cfrac{(a_1+a_3a_4f_1^{(1)})(a_1+a_3f_1^{(1)})}{a_0a_1+a_3f_1^{(1)}},
\end{equation}
with the condition \eqref{eqn:condition_a_qP3D7}, gives $q$-P$_{\rm III}({D_7^{(1)}})$ \eqref{eqn:intro_qP3_1}
by the correspondences \eqref{eqn:intro_correspondence_R0} and \eqref{eqns:para_a_b}.
Moreover, the action of $R_{13}$:
\begin{subequations}
\begin{align}
 &R_{13}(f_1^{(1)})=f_1^{(2)},\\
 &\left(R_{13}(f_1^{(2)})f_1^{(2)}-\cfrac{a_1 a_2}{a_3 a_4}\right)\left(f_1^{(2)}f_1^{(1)}-\cfrac{a_1 a_2}{a_3 a_4}\right)
 =\cfrac{{a_1}^3 a_2}{a_3 {a_4}^2}\,\cfrac{(a_0 a_4 f_1^{(2)}+a_2)(a_4 f_1^{(2)}+a_2)}{a_4 f_1^{(2)}+a_1 a_2},
\end{align}
\end{subequations}
with the condition \eqref{eqn:condition_a_qP4}, gives $q$-P$_{\rm IV}$ \eqref{eqn:intro_qP4_1}
by the correspondences \eqref{eqn:intro_correspondence_R13} and \eqref{eqns:para_a_b}.
\subsection{$\omega$-lattice}
In this section, we define the $\omega$-variables by the ratios of the $\tau$-variables
and then construct the $\omega$-lattice of type $A_4^{(1)}$.

Let us define the fifteen $\omega$-variables by
\begin{equation}\label{eqn:def_omega_A4}
 \omega_1^{(j)}=\cfrac{\tau_1^{(j)}}{\tau_1^{(j+1)}},\quad
 \omega_2^{(j)}=\cfrac{\tau_2^{(j)}}{\tau_1^{(j+2)}},\quad
 \omega_3^{(j)}=\cfrac{\tau_1^{(j-1)}}{\tau_2^{(j-1)}},\quad
 j\in\bbZ/5\bbZ,
\end{equation}
which satisfy
\begin{equation}
 f_1^{(j)}=\cfrac{\omega_2^{(j)}}{\omega_1^{(j)}},\quad
 f_2^{(j)}=\cfrac{a_ja_{j+1}}{{a_{j+3}}^2}\,\cfrac{\omega_1^{(j+1)}(a_{j+2}a_{j+3}\omega_1^{(j+3)}+a_j\omega_2^{(j+3)})}{\omega_1^{(j+3)}\omega_2^{(j+1)}},\quad
 j\in\bbZ/5\bbZ.
\end{equation}
From the definition above and the relations \eqref{eqn:A4_conditions_tau},
they satisfy the following nine relations:
{\allowdisplaybreaks
\begin{subequations}\label{eqn:A4_conditions_omega}
\begin{align}
 &\omega_2^{(5)}=\cfrac{\omega_1^{(1)}\omega_1^{(5)}}{\omega_3^{(1)}},\quad
 \omega_3^{(5)}=\cfrac{\omega_1^{(4)}\omega_1^{(5)}}{\omega_2^{(4)}},\quad
 \omega_2^{(2)}=\cfrac{\omega_1^{(2)}\omega_1^{(3)}}{\omega_3^{(3)}},\quad
 \omega_2^{(3)}=\cfrac{\omega_1^{(3)}\omega_1^{(4)}}{\omega_3^{(4)}},\\
 &\omega_1^{(5)}=\cfrac{1}{\omega_1^{(1)}\omega_1^{(2)}\omega_1^{(3)}\omega_1^{(4)}},\quad
 \omega_1^{(2)}=\cfrac{\omega_2^{(1)}\omega_3^{(2)}}{\omega_1^{(1)}},\\
 &\omega_1^{(3)}=\cfrac{a_2\omega_3^{(3)}(a_1a_2\omega_1^{(1)}\omega_2^{(4)}-a_0a_4\omega_1^{(4)}\omega_3^{(1)})}{a_0{a_4}^2\omega_1^{(4)}\omega_3^{(1)}},\\
 &\omega_1^{(4)}=\cfrac{a_3\omega_3^{(4)}(a_2a_3\omega_2^{(1)}-a_0a_1\omega_3^{(1)})}{a_1{a_0}^2\omega_3^{(1)}},\quad
 \omega_2^{(4)}=\cfrac{a_3a_4\omega_3^{(4)}(a_1\omega_1^{(1)}+a_3\omega_2^{(1)})}{a_0{a_1}^2\omega_1^{(1)}}.
\end{align}
\end{subequations}
By inspection, we see that there are six essential $\omega$-variables.
The action of $\widetilde{W}(A_4^{(1)})$ on the $\omega$-variables is given by the lemma below, 
which follows from the action \eqref{eqns:A4_weylaction_tau} and the definition \eqref{eqn:def_omega_A4}.
}
\begin{lemma}
The action of $\widetilde{W}(A_4^{(1)})$ on the fifteen $\omega$-variables is given by
{\allowdisplaybreaks
\begin{subequations}\label{eqns:actions_A4_omega}
\begin{align}
 &s_j(\omega_1^{(j+4)})=\omega_3^{(j)},\quad
 s_j(\omega_1^{(j)})=\omega_2^{(j+4)},\quad
 s_j(\omega_3^{(j)})=\omega_1^{(j+4)},\quad
 s_j(\omega_2^{(j+4)})=\omega_1^{(j)},\\
 &s_j(\omega_2^{(j+3)})
 =\cfrac{a_{j+3}a_{j+4}\omega_1^{(j+3)}\omega_1^{(j+4)}(a_ja_{j+1}\omega_1^{(j+1)}+a_3\omega_2^{(j+1)})}{{a_{j+1}}^2\omega_1^{(j+1)}\omega_2^{(j+4)}},\\
 &s_j(\omega_2^{(j)})
 =\cfrac{a_{j+4}\omega_1^{(j+4)}(a_4a_0\omega_2^{(j+2)}+a_2\omega_1^{(j+2)})}{a_ja_{j+1}{a_{j+2}}^2\omega_1^{(j+2)}},\\
 &s_j(\omega_3^{(j+1)})
 =\cfrac{a_ja_{j+1}{a_{j+2}}^2\omega_1^{(j+1)}\omega_1^{(j+2)}\omega_2^{(j+4)}}{a_{j+4}\omega_1^{(j+4)}(a_{j+2}\omega_1^{(j+2)}+a_ja_{j+4}\omega_2^{(j+2)})},\\
 &s_j(\omega_3^{(j+4)})
 =\cfrac{{a_{j+1}}^2\omega_1^{(j+1)}\omega_1^{(j)}}{a_{j+3}a_{j+4}(a_ja_{j+1}\omega_1^{(j+1)}+a_{j+3}\omega_2^{(j+1)})},\\
 &\sigma(\omega_1^{(j)})=\omega_1^{(j+1)},\quad
 \sigma(\omega_2^{(j)})=\omega_2^{(j+1)},\quad
 \sigma(\omega_3^{(j)})=\omega_3^{(j+1)},\\
 &\iota(\omega_1^{(j)})=\cfrac{1}{\omega_1^{(4-j)}},\quad
 \iota(\omega_2^{(j)})=\cfrac{1}{\omega_3^{(4-j)}},\quad
 \iota(\omega_3^{(j)})=\cfrac{1}{\omega_2^{(4-j)}},
\end{align}
\end{subequations}
where $j\in\bbZ/5\bbZ$.
}
\end{lemma}

We define $\omega$-functions by
\begin{equation}\label{eqn:A4_omega_function}
 \omega_{l_0,l_1,l_2,l_3,l_4}^{(j)}={T_0}^{l_0}{T_1}^{l_1}{T_2}^{l_2}{T_3}^{l_3}{T_4}^{l_4}(\omega_3^{(j)}),
\end{equation}
where $j=1,\dots,5$ and $l_0,\dots,l_4\in\bbZ$.
We note that 
\begin{subequations}
\begin{align}
 &\omega_1^{(1)}=\omega_{1,0,0,0,0}^{(1)},\quad
 \omega_2^{(1)}=\omega_{1,0,1,0,0}^{(1)},\quad
 \omega_3^{(1)}=\omega_{0,0,0,0,0}^{(1)},\\
 &\omega_1^{(2)}=\omega_{0,1,0,0,0}^{(2)},\quad
 \omega_2^{(2)}=\omega_{0,1,0,1,0}^{(2)},\quad
 \omega_3^{(2)}=\omega_{0,0,0,0,0}^{(2)},\\
  &\omega_1^{(3)}=\omega_{0,0,1,0,0}^{(3)},\quad
 \omega_2^{(3)}=\omega_{0,0,1,0,1}^{(3)},\quad
 \omega_3^{(3)}=\omega_{0,0,0,0,0}^{(3)},\\
 &\omega_1^{(4)}=\omega_{0,0,0,1,0}^{(4)},\quad
 \omega_2^{(4)}=\omega_{1,0,0,1,0}^{(4)},\quad
 \omega_3^{(4)}=\omega_{0,0,0,0,0}^{(4)},\\
 &\omega_1^{(5)}=\omega_{0,0,0,0,1}^{(5)},\quad
 \omega_2^{(5)}=\omega_{0,1,0,0,1}^{(5)},\quad
 \omega_3^{(5)}=\omega_{0,0,0,0,0}^{(5)}.
\end{align}
\end{subequations}

Now we are in a position to construct the $\omega$-lattice of type $A_4^{(1)}$. 
Let us consider the following lattice (see Figure \ref{fig:omega_latticeA4}):
\begin{equation}\label{eqn:omega_lattice_A4}
 \sum_{i=0}^4l_i\bmv_i\subset\bbZ^5,
\end{equation}
whose vertices $\bmv_i$, $i=0,\dots,4$, are defined by
\begin{subequations}\label{eqns:vectors_v_A4}
\begin{align}
 &\bmv_0=(-1,-1,-1,-1,4),\quad
 \bmv_1=(4,-1,-1,-1,-1),\quad
 \bmv_2=(-1,4,-1,-1,-1),\\
 &\bmv_3=(-1,-1,4,-1,-1),\quad
 \bmv_4=(-1,-1,-1,4,-1),
\end{align}
\end{subequations}
and satisfy $\bmv_0+\bmv_1+\bmv_2+\bmv_3+\bmv_4=\bm{0}$.
For simplicity, we here use the following notation:
\begin{equation}\label{eqns:notations_v_A4}
 \bmv_{k_1\dots k_n}=\sum_{i=1}^n \bmv_{k_i},\quad
 k_i\in\{0,\dots,4\}.
\end{equation}
Let us assign the $\tau$ functions $\tau_{l_0,l_1,l_2,l_3,l_4}$ and the $\omega$-functions $\omega_{l_0,l_1,l_2,l_3,l_4}^{(j)}$
to the vertices and the edges of the lattice \eqref{eqn:omega_lattice_A4} by the following correspondence:
\begin{subequations}
\begin{align}
 \tau_{l_0,l_1,l_2,l_3,l_4}
 &\leftrightarrow
 \bml+\bmv_{124},\\
 \omega_{l_0,l_1,l_2,l_3,l_4}^{(1)}
 &\leftrightarrow
 {\rm edge}(\bml+\bmv_{1234};1),\\
 \omega_{l_0,l_1,l_2,l_3,l_4}^{(2)}
 &\leftrightarrow
 {\rm edge}(\bml;2),\\
 \omega_{l_0,l_1,l_2,l_3,l_4}^{(3)}
 &\leftrightarrow
 {\rm edge}(\bml+\bmv_1;3),\\
 \omega_{l_0,l_1,l_2,l_3,l_4}^{(4)}
 &\leftrightarrow
 {\rm edge}(\bml+\bmv_{12};4),\\
 \omega_{l_0,l_1,l_2,l_3,l_4}^{(5)}
 &\leftrightarrow
 {\rm edge}(\bml+\bmv_{123};0),
\end{align}
\end{subequations}
where $\bml=\sum_{i=0}^4l_i\bmv_i$. 
Here, ${\rm edge}({\bm A};i)$ is a edge connecting a vertex ${\bm A}$ to a vertex $({\bm A}+\bmv_i)$.
We refer to the lattice \eqref{eqn:omega_lattice_A4} with the $\omega$-functions $\omega_{l_0,l_1,l_2,l_3,l_4}^{(j)}$
as $\omega$-lattice of type $A_4^{(1)}$.
We note that the configurations of the $\tau$-variables on the $\omega$-lattice are given by
\begin{subequations}
\begin{align}
 &(\tau_1^{(1)},\tau_1^{(2)},\tau_1^{(3)},\tau_1^{(4)},\tau_1^{(5)})
 \leftrightarrow
 ({\bm 0},\bmv_1,\bmv_{12},\bmv_{123},\bmv_{1234}),\\
 &(\tau_2^{(1)},\tau_2^{(2)},\tau_2^{(3)},\tau_2^{(4)},\tau_2^{(5)})
 \leftrightarrow
 (\bmv_2,\bmv_{13},\bmv_{124},\bmv_{0123},\bmv_{11234}),
\end{align}
\end{subequations}
while those of the $\omega$-variables are given by
\begin{subequations}
\begin{align}
 &(\omega_1^{(1)},\omega_2^{(1)},\omega_3^{(1)})
 \leftrightarrow
 \Big({\rm edge}(\bm{0};1),{\rm edge}(\bmv_2;1),{\rm edge}(\bmv_{1234};1)\Big),\\
 &(\omega_1^{(2)},\omega_2^{(2)},\omega_3^{(2)})
 \leftrightarrow
 \Big({\rm edge}(\bmv_1;2),{\rm edge}(\bmv_{13};2),{\rm edge}(\bm{0};2)\Big),\\
 &(\omega_1^{(3)},\omega_2^{(3)},\omega_3^{(3)})
 \leftrightarrow
 \Big({\rm edge}(\bmv_{12};3),{\rm edge}(\bmv_{124};3),{\rm edge}(\bmv_1;3)\Big),\\
 &(\omega_1^{(4)},\omega_2^{(4)},\omega_3^{(4)})
 \leftrightarrow
 \Big({\rm edge}(\bmv_{123};4),{\rm edge}(\bmv_{0123};4),{\rm edge}(\bmv_{12};4)\Big),\\
 &(\omega_1^{(5)},\omega_2^{(5)},\omega_3^{(5)})
 \leftrightarrow
 \Big({\rm edge}(\bmv_{1234};0),{\rm edge}(\bmv_{11234};0),{\rm edge}(\bmv_{123};0)\Big).
\end{align}
\end{subequations}

\begin{figure}[t]
\begin{center}
\includegraphics[width=0.8\textwidth]{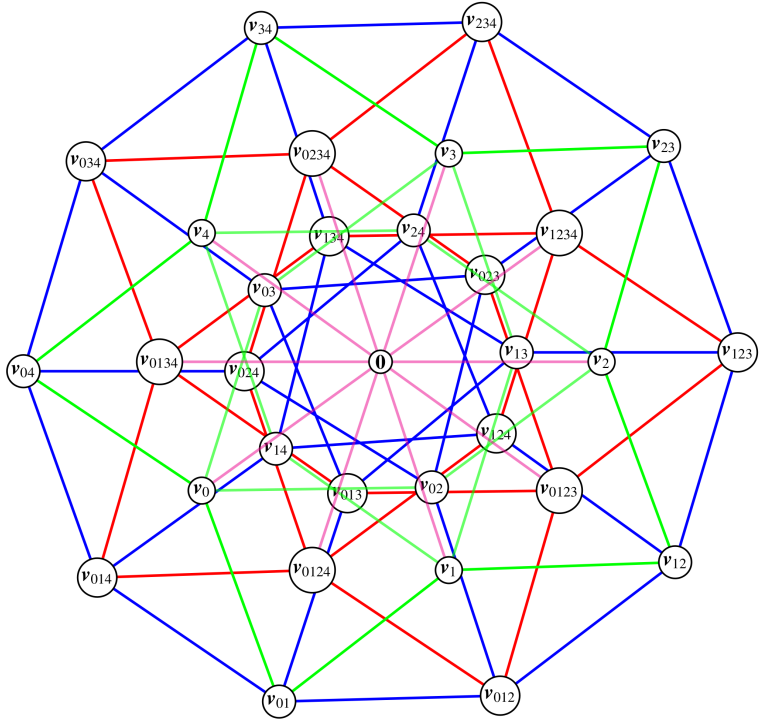}
\end{center}
\caption{The lattice \eqref{eqn:omega_lattice_A4} around the origin, 
which is a 2-dimensional projection of the Voronoi cell of type $A_4$.
Refer to \eqref{eqns:vectors_v_A4} and \eqref{eqns:notations_v_A4} for $\bmv$.
The directions from $\bm{0}$ to $\bmv_i$, $i=0,\dots,4$, correspond to the $T_i$-directions, $i=0,\dots,4$, respectively.
}
\label{fig:omega_latticeA4}
\end{figure}

On the $\omega$-lattice various quad-equations of ABS-type can be derived, e.g.
\begin{subequations}\label{eqns:A4_H3}
\begin{align}
 &\cfrac{T_0T_2(\omega_3^{(1)})}{\omega_3^{(1)}}
 =\cfrac{a_0}{a_2 a_3}\,
 \cfrac{(1-a_2) {a_1}^2 a_2T_0(\omega_3^{(1)})+(1-a_1) T_2(\omega_3^{(1)})}
 {a_1 a_2 T_0(\omega_3^{(1)})-T_2(\omega_3^{(1)})},\\
 &\cfrac{T_0T_3(\omega_3^{(1)})}{\omega_3^{(1)}}
 =\cfrac{a_0}{a_2}\,
 \cfrac{{a_1}^2a_2 a_3 (1-a_2 a_3) T_0(\omega_3^{(1)})+(1-a_1) T_3(\omega_3^{(1)})}
 {a_1 a_2 a_3 T_0(\omega_3^{(1)})-T_3(\omega_3^{(1)})},\\
 &\cfrac{T_0T_4(\omega_3^{(1)})}{\omega_3^{(1)}}
 =\cfrac{a_0 a_4}{a_2}\,
 \cfrac{{a_1}^2 a_2 a_3 a_4 (1-a_2 a_3 a_4) T_0(\omega_3^{(1)})+(1-a_1) T_4(\omega_3^{(1)})}
 {a_1 a_2 a_3 a_4 T_0(\omega_3^{(1)})-T_4(\omega_3^{(1)})},\\
 &\cfrac{T_2T_3(\omega_3^{(1)})}{\omega_3^{(1)}}
 =a_0 {a_1}^2\,
 \cfrac{a_3 (1-a_2 a_3) T_2(\omega_3^{(1)})-(1-a_2) T_3(\omega_3^{(1)})}{a_3 T_2(\omega_3^{(1)})-T_3(\omega_3^{(1)})},\\
 &\cfrac{T_2T_4(\omega_3^{(1)})}{\omega_3^{(1)}}
 =a_0 {a_1}^2 a_4\,
 \cfrac{a_3 a_4 (1-a_2 a_3 a_4) T_2(\omega_3^{(1)})-(1-a_2) T_4(\omega_3^{(1)})}{a_3 a_4 T_2(\omega_3^{(1)})-T_4(\omega_3^{(1)})},\\
 &\cfrac{T_3T_4(\omega_3^{(1)})}{\omega_3^{(1)}}
 =a_0 {a_1}^2 a_3 a_4\,
 \cfrac{a_4 (1-a_2 a_3 a_4) T_3(\omega_3^{(1)})-(1-a_2 a_3) T_4(\omega_3^{(1)})}{a_4 T_3(\omega_3^{(1)})-T_4(\omega_3^{(1)})},
\end{align}
\end{subequations}
\begin{subequations}\label{eqns:A4_D4}
\begin{align}
 &\cfrac{\omega_1^{(1)}}{\omega_3^{(1)}}
 -\cfrac{{a_0}^2 a_4}{{a_2}^2 a_3}\,\cfrac{\omega_2^{(3)}}{\omega_3^{(3)}}=-\cfrac{a_0}{a_2},\\
 &\cfrac{\omega_1^{(1)}}{\omega_3^{(1)}}
 -\cfrac{{a_0}^2 a_4}{{a_2}^2}\,\cfrac{T_4(\omega_2^{(2)})}{\omega_1^{(2)}}=-\cfrac{a_0}{a_2 a_3},\\
 &\cfrac{\omega_1^{(1)}}{\omega_3^{(1)}}
 -\cfrac{{a_0}^2}{{a_2}^2 a_3}\,\cfrac{\omega_1^{(4)}}{{T_2}^{-1}T_3^{-1}(\omega_1^{(4)})}=-\cfrac{a_0 a_4}{a_2}.
\end{align}
\end{subequations}
Note that Equations \eqref{eqns:A4_H3} are relations between the $\omega$-function $\omega_{l_0,l_1,l_2,l_3,l_4}^{(1)}$,
but Equations \eqref{eqns:A4_D4} are the relations between 
$\omega_{l_0,l_1,l_2,l_3,l_4}^{(1)}$ and $\omega_{l_0,l_1,l_2,l_3,l_4}^{(3)}$,
$\omega_{l_0,l_1,l_2,l_3,l_4}^{(1)}$ and $\omega_{l_0,l_1,l_2,l_3,l_4}^{(2)}$ and
$\omega_{l_0,l_1,l_2,l_3,l_4}^{(1)}$ and $\omega_{l_0,l_1,l_2,l_3,l_4}^{(4)}$, respectively.
Each equation of Equations \eqref{eqns:A4_H3} and that of Equations \eqref{eqns:A4_D4}
are of $H3$- and $D4$-types in the ABS classification \cite{ABS2003:MR1962121,ABS2009:MR2503862,BollR2011:MR2846098,BollR2012:MR3010833,BollR:thesis},
respectively.
Details of the $\omega$-lattice  of type $A_4^{(1)}$ will be discussed in a forthcoming paper (N. Joshi, N. Nakazono and Y. Shi, in preparation).

\section{Construction of the lattice $\omega_{A_2+A_1}$}
\label{section:affineWeylA2}
In this section, we consider the extended affine Weyl group $\widetilde{W}((A_2\rtimes A_1)^{(1)})$
given by the following six generators:
\begin{equation}\label{eqn:WA2_elements}
 w_0=s_0,\quad
 w_1=s_1s_2s_1,\quad
 w_2=s_3s_4s_3,\quad
 r_0=\iota,\quad
 r_1=\sigma \iota s_2s_4,\quad
 \pi=\sigma^3 \iota s_4.
\end{equation}
The details of $\widetilde{W}((A_2\rtimes A_1)^{(1)})$ is discussed in Appendix \ref{section:Q(A2A1)andW(A2)}.
Using this group, we construct another important lattice $\omega_{A_2+A_1}$.
Moreover, we show that the $q$-Painlev\'e equations \eqref{eqns:intro_qPs} can be derived also as the relations on the lattice $\omega_{A_2+A_1}$.
\subsection{Affine Weyl group $\widetilde{W}((A_2\rtimes A_1)^{(1)})$}
In this section, we consider the birational action of $\widetilde{W}((A_2\rtimes A_1)^{(1)})$
on the parameters $b_i$, $i=0,1,2,3$, and $p$ defined by \eqref{eqns:para_a_b}
and on the particular $\omega$-variables $\omega_i^{(j)}$, $i=1,2,3$ and $j=1,3$, given by \eqref{eqn:def_omega_A4}.
We note that from the relations \eqref{eqn:A4_conditions_omega}, these six $\omega$-variables satisfy the following two relations:
\begin{subequations}\label{eqns:condition_omega_A2}
\begin{align}
 &\omega_2^{(1)}=\cfrac{b_0{b_3}^{3/2} \omega_3^{(1)}(p^2{b_0}^3 {b_3}^{1/2}\omega_2^{(3)}-b_1 b_2 \omega_1^{(3)})}{{b_1}^2 \omega_1^{(3)}},
 \label{eqn:condition_omega_A2_1}\\
 &\omega_3^{(3)}=\cfrac{p^2{b_0}^2 {b_1}^2 {b_3}^2 \omega_2^{(3)}\omega_3^{(1)}}{\omega_1^{(1)}-pb_0 b_1 b_2 {b_3}^{1/2}\omega_3^{(1)}}.
 \label{eqn:condition_omega_A2_2} 
\end{align}
\end{subequations}
Therefore, essential $\omega$-variables used here are four.
The action of $\widetilde{W}((A_2\rtimes A_1)^{(1)})$ on the parameters is given by
\begin{subequations}
\begin{align}
 &w_0:(b_0,b_1,b_2,b_3,p)\mapsto ({b_0}^{-1} p^{-2},p^{-1}{b_0}^{-1}b_1,p^{-1}{b_0}^{-1}b_2,b_3,p),\\
 &w_1:(b_0,b_1,b_2,b_3,p)\mapsto (b_0{b_1}^{-1},{b_1}^{-1},{b_1}^{-1}b_2,b_3,p),\\
 &w_2:(b_0,b_1,b_2,b_3,p)\mapsto (b_1,b_0,b_2,b_3,p),\\
 &r_0:(b_0,b_1,b_2,b_3,p)\mapsto ({b_0}^{-1},{b_0}^{-1}b_1,p^{-1}{b_0}^{-1}b_2,{b_3}^{-1},p^{-1}),\\
 &r_1:(b_0,b_1,b_2,b_3,p)\mapsto (pb_1,pb_0,p^{-1}b_2,{b_3}^{-1},p^{-1}),\\
 &\pi:(b_0,b_1,b_2,b_3,p)\mapsto (p b_0{b_1}^{-1},{b_1}^{-1},p^{-1}{b_1}^{-1}b_2,b_3,p^{-1}),
\end{align}
\end{subequations}
while that on the six $\omega$-variables is given by
{\allowdisplaybreaks
\begin{subequations}
\begin{align}
 &w_0(\omega_3^{(1)})
 =\cfrac{p \omega_1^{(3)} \omega_3^{(3)} (p{b_1}^2 \omega_2^{(1)}+b_3 \omega_1^{(1)})}
 {{b_3}^2 \omega_2^{(3)}(p {b_1}^2 b_3 \omega_1^{(3)}+\omega_3^{(3)})},\\
 &w_0(\omega_2^{(3)})
 =\cfrac{p^3{b_0}^2 \omega_2^{(3)} (p{b_0}^2 b_3 \omega_1^{(1)}+{b_1}^2 \omega_2^{(1)})}{p{b_1}^2 \omega_2^{(1)}+b_3 \omega_1^{(1)}},\\
 &w_1(\omega_1^{(1)})
 =\cfrac{\omega_1^{(1)}(p{b_0}^2b_3 \omega_2^{(3)}+\omega_1^{(3)})}{{b_1}^2(p{b_0}^2 b_3 \omega_2^{(3)}+{b_1}^2 \omega_1^{(3)})},\\
 &w_1(\omega_3^{(3)})
 =\cfrac{p b_3 \omega_3^{(1)} \omega_3^{(3)}(p{b_0}^2 b_3 \omega_2^{(3)}+{b_1}^2 \omega_1^{(3)})}
 {b_1(-pb_0 b_2 {b_3}^{1/2}\omega_3^{(1)} \omega_3^{(3)}+p{b_1}^3 b_3 \omega_1^{(3)} \omega_3^{(1)}+b_1 \omega_1^{(1)} \omega_3^{(3)})},\\
 &w_2(\omega_2^{(1)})
 =\cfrac{b_1 b_3 \omega_3^{(1)}(p^2{b_0}^2 b_1 b_3 \omega_1^{(1)} \omega_2^{(3)}+p {b_1}^3 \omega_1^{(3)} \omega_2^{(1)}-b_0 b_2{b_3}^{1/2} \omega_1^{(1)} \omega_1^{(3)})}{{b_0}^2 \omega_1^{(3)} (p {b_1}^2 b_3 \omega_3^{(1)}+\omega_1^{(1)})},\\
 &w_2(\omega_1^{(3)})
 =\cfrac{{b_1}^2 \omega_1^{(3)}(p {b_1}^2 b_3 \omega_3^{(1)}+\omega_1^{(1)})}{{b_0}^2 (p{b_0}^2b_3 \omega_3^{(1)}+\omega_1^{(1)})},\quad
 r_0(\omega_i^{(1)})=\cfrac{1}{\omega_{-i+2}^{(3)}},\quad
 r_0(\omega_i^{(3)})=\cfrac{1}{\omega_{-i+2}^{(1)}},\\
 &r_1(\omega_1^{(1)})=\cfrac{1}{\omega_2^{(3)}},\quad
 r_1(\omega_2^{(1)})=\cfrac{1}{\omega_1^{(3)}},\quad
 r_1(\omega_3^{(1)})
 =\cfrac{-p {b_0}^2 {b_3}^2 \omega_1^{(1)}}{b_1 \omega_1^{(3)} (b_0 b_2 {b_3}^{1/2} \omega_1^{(1)}-p{b_1}^3 \omega_2^{(1)})},\\
 &r_1(\omega_1^{(3)})=\cfrac{1}{\omega_2^{(1)}},\quad
 r_1(\omega_2^{(3)})=\cfrac{1}{\omega_1^{(1)}},\\
 &r_1(\omega_3^{(3)})
 =\cfrac{-p^3{b_0}^4 b_1\omega_2^{(3)}}
 {p{b_0}^3 b_2 b_3^{3/2} \omega_1^{(1)} \omega_2^{(3)}+b_0 b_2 {b_3}^{1/2} \omega_1^{(1)} \omega_1^{(3)}-p{b_1}^3  \omega_2^{(1)} \omega_1^{(3)}},\\
 &\pi(\omega_i^{(1)})=\cfrac{1}{\omega_{-i+2}^{(1)}},\quad
 \pi(\omega_1^{(3)})=\cfrac{1}{\omega_2^{(3)}},\quad
 \pi(\omega_2^{(3)})=\cfrac{1}{\omega_1^{(3)}},\\
 &\pi(\omega_3^{(3)})
 =\cfrac{-p{b_0}^2 {b_3}^2\omega_1^{(1)}}{b_1 \omega_1^{(3)}(b_0 b_2 {b_3}^{1/2} \omega_1^{(1)}-p{b_1}^3 \omega_2^{(1)})},
\end{align}
\end{subequations}
where $i\in\bbZ/3\bbZ$,
which follow from \eqref{eqn:A4_weylaction_a}, \eqref{eqns:para_a_b}, \eqref{eqns:actions_A4_omega} and \eqref{eqn:WA2_elements}.
}

Let
\begin{equation}\label{eqn:def_rho}
 \rho_1=\pi r_0 w_1w_2,\quad
 \rho_2=\pi r_0 w_0w_1,\quad
 \rho_3=\pi r_0 w_2w_0,\quad
 \rho_4=\pi r_1 r_0 r_1.
\end{equation}
Note here that the transformations $\rho_i$, $i =1,\dots,4$, are translations on the root system $Q((A_2+A_1)^{(1)})$ \eqref{eqn:root_symmetry_A2A1} 
(see Appendix \ref{section:Q(A2A1)andW(A2)} for details).
The translations $\rho_i$, $i =1,\dots,4$, commute with each other and 
\begin{equation}
 \rho_1\rho_2\rho_3\rho_4=1.
\end{equation}
Their actions on the parameters are given by
\begin{subequations}
\begin{align}
 &\rho_1:(b_0,b_1,b_2,b_3)\mapsto(pb_0,b_1,b_2,{b_3}^{-1}),\\
 &\rho_2:(b_0,b_1,b_2,b_3)\mapsto(b_0,pb_1,b_2,{b_3}^{-1}),\\
 &\rho_3:(b_0,b_1,b_2,b_3)\mapsto(p^{-1}b_0,p^{-1}b_1,p^{-1}b_2,{b_3}^{-1}),\\
 &\rho_4:(b_0,b_1,b_2,b_3)\mapsto(b_0,b_1,pb_2,{b_3}^{-1}),
\end{align}
\end{subequations}
where $p$ is invariant under their actions.

\subsection{Lattice $\omega_{A_2+A_1}$}
In this section, we define the $\omega$-functions associated with the translations on the root system $Q((A_2+A_1)^{(1)})$
and then construct the lattice $\omega_{A_2+A_1}$.

We define $\omega$-functions by using the translations $\rho_i$, $i=1,\dots,4$, as follows:
\begin{equation}\label{eqn:A5_omegafun}
 \omega_{l_1,l_2,l_3,l_4}={\rho_1}^{l_1}{\rho_2}^{l_2}{\rho_3}^{l_3}{\rho_4}^{l_4}(\omega_3^{(1)}),
\end{equation}
where $l_i\in\mathbb{Z}$.
We note that
\begin{subequations}\label{eqns:notation_omega_A2}
\begin{align}
 &\omega_1^{(1)}=\omega_{1,1,0,0},\quad
 \omega_2^{(1)}=\omega_{2,1,1,0},\quad
 \omega_3^{(1)}=\omega_{0,0,0,0},\\
 &\omega_1^{(3)}=\omega_{1,0,0,0},\quad
 \omega_2^{(3)}=\omega_{1,1,1,0},\quad
 \omega_3^{(3)}=\omega_{1,1,0,1}.
\end{align}
\end{subequations}
Let us assign the $\omega$-functions $\omega_{l_1,l_2,l_3,l_4}$ to the vertices of the lattice
\begin{equation}\label{eqn:omega_lattice_A2A1}
 \left.\left\{\sum_{i=1}^4l_i\bmv_i\,\right|\, l_1,\dots,l_4\in\bbZ\right\}
\end{equation}
by the following correspondence:
\begin{equation}\label{eqn:omega_R3}
 \omega_{l_1,l_2,l_3,l_4}
 \leftrightarrow
 l_1\bmv_1+l_2\bmv_2+l_3\bmv_3+l_4\bmv_4.
\end{equation}
Here, $\bmv_i$, $i=1,\dots,4$, are defined by
\begin{equation}\label{eqn:vectors_v_A2A1}
 \bmv_1=(1,1,1),\quad
 \bmv_2=(-1,-1,1),\quad
 \bmv_3=(1,-1,-1),\quad
 \bmv_4=(-1,1,-1),
\end{equation}
and satisfy $\bmv_1+\bmv_2+\bmv_3+\bmv_4=\bm{0}$.
We here refer to the lattice \eqref{eqn:omega_lattice_A2A1}
with the $\omega$-functions $\omega_{l_1,l_2,l_3,l_4}$ as lattice $\omega_{A_2+A_1}$.
We note that the configurations of the $\omega$-variables on the lattice $\omega_{A_2+A_1}$ are given by
\begin{equation}
 (\omega_1^{(1)},\omega_2^{(1)},\omega_3^{(1)},\omega_1^{(3)},\omega_2^{(3)},\omega_3^{(3)})
 \leftrightarrow
 (\bmv_1+\bmv_2,\bmv_1-\bmv_4,\bm{0},\bmv_1,-\bmv_4,-\bmv_3).
\end{equation}
See the example given in Figure \ref{fig:quad} to see the quadrilateral associated with 
$\omega_1^{(1)}$, $\omega_3^{(1)}$, $\omega_2^{(3)}$ and $\omega_3^{(3)}$.

\begin{figure}[t]
\begin{center}
\includegraphics[width=0.38\textwidth]{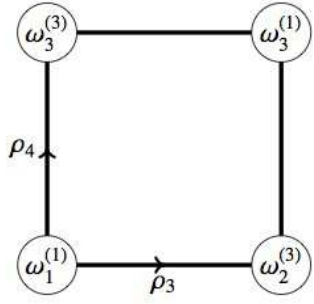}
\end{center}
\caption{A quadrilateral associated with the $\omega$-variables $\omega_1^{(1)}$, $\omega_3^{(1)}$, $\omega_2^{(3)}$ and $\omega_3^{(3)}$.}
\label{fig:quad}
\end{figure}

The 14 vertices around $\bml\in\omega_{A_2+A_1}$:
\begin{equation}
 \{\bml\pm\bmv_i,~\bml+\bmv_i+\bmv_j\,|\,i,j=1,\dots,4,~i\neq j\},
\end{equation}
collectively forms the rhombic dodecahedron (see Figure \ref{fig:Rhombic_dodecahedron}).
Letting $\bar{V}(\bml)$ be the rhombic dodecahedron with the center $\bml\in\omega_{A_2+A_1}$:
\begin{equation}
 \bar{V}(\bml)=\{\bml\}\cup V(\bml),
\end{equation}
then the following holds:
\begin{equation}
 \omega_{A_2+A_1}=\bigcup_{\bml\in\omega_{A_2+A_1}}\bar{V}(\bml).
\end{equation}

\begin{figure}[t]
\begin{center}
\includegraphics[width=0.7\textwidth]{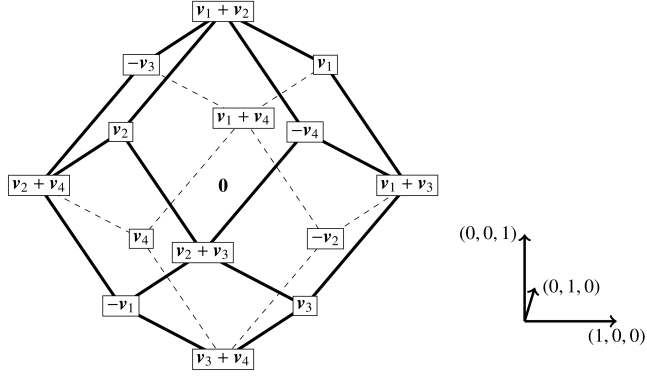}
\end{center}
\caption{The rhombic dodecahedron around $\bm{0}$.
Refer to \eqref{eqn:vectors_v_A2A1} for $\bmv$.
The directions from $\bm{0}$ to $\bmv_i$, $i=1,\dots,4$, correspond to the $\rho_i$-directions, $i=1,\dots,4$, respectively.}
\label{fig:Rhombic_dodecahedron}
\end{figure}

Henceforth, let us consider the quad-equations appearing on the lattice $\omega_{A_2+A_1}$.

\begin{lemma}
The following quad-equations hold on the lattice $\omega_{A_2+A_1}$$:$
{\allowdisplaybreaks
\begin{subequations}\label{eqns:quad_omega_A2}
\begin{align}
 &\cfrac{\omega_{l_1+1,l_2+1,l_3,l_4}}{\omega_{l_1,l_2,l_3,l_4}}
 =p^{2(l_2-l_3)+1}{b_1}^2\la_{l_1+l_2+l_3+l_4}\,
 \cfrac{\omega_{l_1+1,l_2,l_3,l_4}-p^{4(l_1-l_2)}{b_0}^4{b_1}^{-4}\omega_{l_1,l_2+1,l_3,l_4}}
 {p^{2(l_1-l_2)}{b_0}^2{b_1}^{-2}\omega_{l_1,l_2+1,l_3,l_4}-\omega_{l_1+1,l_2,l_3,l_4}},
 \label{eqn:quad_omega_A2_12}\\
 &\cfrac{\omega_{l_1,l_2+1,l_3+1,l_4}}{\omega_{l_1,l_2,l_3,l_4}}
 =\cfrac{p^{2(-l_1+l_3)+1}\la_{l_1+l_2+l_3+l_4}}{{b_0}^2}~
 \cfrac{\omega_{l_1,l_2+1,l_3,l_4}-p^{4(l_2-l_3)}{b_1}^4\omega_{l_1,l_2,l_3+1,l_4}}
 {p^{2(l_2-l_3)}{b_1}^2\omega_{l_1,l_2,l_3+1,l_4}-\omega_{l_1,l_2+1,l_3,l_4}},
 \label{eqn:quad_omega_A2_23}\\
 &\cfrac{\omega_{l_1+1,l_2,l_3+1,l_4}}{\omega_{l_1,l_2,l_3,l_4}}
 =\cfrac{p^{2(-l_2+l_3)+1}\la_{l_1+l_2+l_3+l_4}}{{b_1}^2}~
 \cfrac{\omega_{l_1+1,l_2,l_3,l_4}-p^{4(l_1-l_3)}{b_0}^4\omega_{l_1,l_2,l_3+1,l_4}}
 {p^{2(l_1-l_3)}{b_0}^2\omega_{l_1,l_2,l_3+1,l_4}-\omega_{l_1+1,l_2,l_3,l_4}},
 \label{eqn:quad_omega_A2_13}\\
 &\cfrac{\omega_{l_1,l_2,l_3,l_4}}{\omega_{l_1+1,l_2,l_3,l_4+1}}
 =\left(\cfrac{p^{-2l_1+l_2+l_3-1}b_1 \la_{l_1+l_2+l_3+l_4}}{{b_0}^2}\right)^2
 \cfrac{\omega_{l_1+1,l_2,l_3,l_4}}{\omega_{l_1,l_2,l_3,l_4+1}}
 +\cfrac{p^{-3l_1+l_2+l_3+l_4-1}b_1 b_2 {\la_{l_1+l_2+l_3+l_4}}^{1/2}}{{b_0}^3},
 \label{eqn:quad_omega_A2_14}\\
 &\cfrac{\omega_{l_1,l_2,l_3,l_4}}{\omega_{l_1,l_2+1,l_3,l_4+1}}
 =\left(\cfrac{p^{l_1-2l_2+l_3-1}b_0 \la_{l_1+l_2+l_3+l_4}}{{b_1}^2}\right)^2
 \cfrac{\omega_{l_1,l_2+1,l_3,l_4}}{\omega_{l_1,l_2,l_3,l_4+1}}
 +\cfrac{p^{l_1-3l_2+l_3+l_4-1}b_0 b_2 {\la_{l_1+l_2+l_3+l_4}}^{1/2}}{{b_1}^3},
 \label{eqn:quad_omega_A2_24}\\
 &\cfrac{\omega_{l_1,l_2,l_3,l_4}}{\omega_{l_1,l_2,l_3+1,l_4+1}}
 =\left(p^{l_1+l_2-2l_3-1}b_0 b_1 \la_{l_1+l_2+l_3+l_4}\right)^2
 \cfrac{\omega_{l_1,l_2,l_3+1,l_4}}{\omega_{l_1,l_2,l_3,l_4+1}}
 +p^{l_1+l_2-3l_3+l_4-1}b_0 b_1 b_2 {\la_{l_1+l_2+l_3+l_4}}^{1/2},
 \label{eqn:quad_omega_A2_34}
\end{align}
\end{subequations}
where
\begin{equation}
 \la_l={b_3}^{(-1)^l}.
\end{equation}
Note that each equation of Equations \eqref{eqn:quad_omega_A2_12}--\eqref{eqn:quad_omega_A2_13} 
and that of Equations \eqref{eqn:quad_omega_A2_14}--\eqref{eqn:quad_omega_A2_34}
are of $H3$- and $D4$-types in the ABS classification \cite{ABS2003:MR1962121,ABS2009:MR2503862,BollR2011:MR2846098,BollR2012:MR3010833,BollR:thesis},
respectively.
}
\end{lemma}
\begin{proof}
{\allowdisplaybreaks
Recalling the definitions of $\rho_i$ given in \eqref{eqn:def_rho} and the relations \eqref{eqns:condition_omega_A2}, we have the actions shown below:
\begin{subequations}
\begin{align}
 &\rho_2(\omega_3^{(1)})
 =\cfrac{{b_1}^2 \omega_1^{(3)}(p{b_1}^2 b_3 \omega_3^{(1)}+\omega_1^{(1)})}{{b_0}^2 (p {b_0}^2 b_3 \omega_3^{(1)}+\omega_1^{(1)})},\\
 &\rho_3(\omega_1^{(3)})
 =\cfrac{\omega_1^{(1)} (p {b_0}^2 b_3 \omega_2^{(3)}+\omega_1^{(3)})}{{b_1}^2 (p{b_0}^2b_3 \omega_2^{(3)}+{b_1}^2 \omega_1^{(3)})},\\
 &{\rho_3}^{-1}(\omega_1^{(3)})
 =\cfrac{p^3{b_0}^2 \omega_3^{(1)} (p{b_0}^2 \omega_3^{(3)}+b_3{b_1}^2 \omega_1^{(3)})}{p{b_1}^2b_3 \omega_1^{(3)}+\omega_3^{(3)}},\\
 &{\rho_4}^{-1}(\omega_1^{(3)})
 =\omega_2^{(1)}
 =\cfrac{b_0 {b_3}^{3/2} \omega_3^{(1)} (p^2{b_0}^3 {b_3}^{1/2} \omega_2^{(3)}-b_1 b_2 \omega_1^{(3)})}{{b_1}^2 \omega_1^{(3)}},\\
 &\rho_4(\omega_1^{(3)})
 =\cfrac{{b_0}^2 \omega_1^{(1)} \omega_3^{(3)}}{b_1 {b_3}^{3/2} ({b_1}^3 {b_3}^{1/2} \omega_1^{(3)}-b_0 b_2 \omega_3^{(3)})}.
\end{align}
\end{subequations}
This leads to 
\begin{subequations}
\begin{align}
 &\cfrac{\omega_1^{(1)}}{\omega_3^{(1)}}
 =p{b_1}^2 b_3
 \cfrac{\omega_1^{(3)}-{b_0}^4{b_1}^{-4}\rho_2(\omega_3^{(1)})}
 {{b_0}^2{b_1}^{-2} \rho_2(\omega_3^{(1)})-\omega_1^{(3)}},\\
 &\cfrac{\omega_2^{(3)}}{\omega_1^{(3)}}
 =\cfrac{{b_3}^{-1}}{p{b_0}^2}\,
 \cfrac{\omega_1^{(1)}-{b_1}^4 \rho_3(\omega_1^{(3)})}{{b_1}^2 \rho_3(\omega_1^{(3)})-\omega_1^{(1)}},\\
 &\cfrac{\omega_1^{(3)}}{\omega_3^{(3)}}
 =\cfrac{{b_3}^{-1}}{p{b_1}^2}\,
 \cfrac{{\rho_3}^{-1}(\omega_1^{(3)})-p^4{b_0}^4\omega_3^{(1)}}{p^2{b_0}^2 \omega_3^{(1)}-{\rho_3}^{-1}(\omega_1^{(3)})},\\
 &\cfrac{\omega_2^{(3)}}{\omega_1^{(3)}}
 =\left(\cfrac{b_1 {b_3}^{-1}}{p{b_0}^2}\right)^2\cfrac{{\rho_4}^{-1}(\omega_1^{(3)})}{\omega_3^{(1)}}
 +\cfrac{b_1 b_2{b_3}^{-1/2}}{p^2{b_0}^3},\\
 &\cfrac{\omega_1^{(3)}}{\omega_3^{(3)}}
 =\left(\cfrac{b_0{b_3}^{-1}}{{b_1}^2}\right)^2\cfrac{\omega_1^{(1)}}{\rho_4(\omega_1^{(3)})}
 +\cfrac{b_0 b_2{b_3}^{-1/2}}{{b_1}^3},
\end{align}
\end{subequations}
which in turn lead immediately to Equations \eqref{eqn:quad_omega_A2_12}--\eqref{eqn:quad_omega_A2_24}.
Moreover, we get Equation \eqref{eqn:quad_omega_A2_34} 
from the relation \eqref{eqn:condition_omega_A2_2} or, equivalently,
\begin{equation}
 \cfrac{\omega_1^{(1)}}{\omega_3^{(1)}}
 =(pb_0b_1b_3)^2\cfrac{\omega_2^{(3)}}{\omega_3^{(3)}}
 +pb_0 b_1 b_2 {b_3}^{1/2}.
\end{equation}
Therefore we have completed the proof.
}
\end{proof}

\begin{lemma}
The quad-equations \eqref{eqns:quad_omega_A2} are fundamental relations on the lattice $\omega_{A_2+A_1}$.
\end{lemma}
\begin{proof}
In this proof we will show that any $\omega$-function $\omega_{l_1,l_2,l_3,l_4}$ can be calculated by the quad-equations \eqref{eqns:quad_omega_A2} 
with four initial values: $\omega_1^{(1)}$, $\omega_3^{(1)}$, $\omega_1^{(3)}$ and $\omega_2^{(3)}$
$($or, $\omega_{1,1,0,0}$, $\omega_{0,0,0,0}$, $\omega_{1,0,0,0}$ and $\omega_{1,1,1,0}$$)$.

First, we obtain the values of all $\omega$-functions on $\bar{V}(\bm{0})$ from the initial values by the following steps.
\begin{description}
\item[Step 1]
By using Equations 
\eqref{eqn:quad_omega_A2_12}$_{(0,0,0,0)}$,
\eqref{eqn:quad_omega_A2_23}$_{(1,0,0,0)}$ and
\eqref{eqn:quad_omega_A2_34}$_{(1,1,0,0)}$,
the functions on $\bmv_2$, $\bmv_1+\bmv_3$ and $-\bmv_3$ can be calculated, respectively.
\item[Step 2]
By using Equations
\eqref{eqn:quad_omega_A2_13}$_{(0,0,0,0)}$,
\eqref{eqn:quad_omega_A2_13}$_{(0,1,0,0)}$,
\eqref{eqn:quad_omega_A2_14}$_{(0,1,0,0)}$,
\eqref{eqn:quad_omega_A2_24}$_{(1,0,0,0)}$ and
\eqref{eqn:quad_omega_A2_24}$_{(1,0,1,0)}$,
the functions on $\bmv_3$, $\bmv_2+\bmv_3$, $\bmv_2+\bmv_4$, $\bmv_1+\bmv_4$ and $-\bmv_2$
can be calculated, respectively.
\item[Step 3]
By using Equations 
\eqref{eqn:quad_omega_A2_12}$_{(0,0,0,1)}$,
\eqref{eqn:quad_omega_A2_14}$_{(0,0,1,0)}$ and
\eqref{eqn:quad_omega_A2_34}$_{(0,1,0,0)}$,
the functions on $\bmv_4$, $\bmv_3+\bmv_4$ and $-\bmv_1$
can be calculated, respectively.
\end{description} 
Note that the subscripts of the equation numbers $(l_1,l_2,l_3,l_4)$ denote the values of the parameters $l_i$, $i=1,\dots,4$, in the equations.

Next, we consider $\bar{V}(\bmv_1)$.
From the determined $\omega$-functions on
\begin{equation}
 \bar{V}(\bm{0})\cap \bar{V}(\bmv_1)=\{\bm{0},~\bmv_1,~\bmv_1+\bmv_i,~-\bmv_i\,|\, i=2,3,4\},
\end{equation}
we can obtain the values of the $\omega$-functions on
\begin{equation}
 \bar{V}(\bmv_1)-\bar{V}(\bm{0})=\{\bmv_1-\bmv_i,~ 2\bmv_1,~2\bmv_1+\bmv_i\,|\, i=2,3,4\},
\end{equation}
by the following steps.
\begin{description}
\item[Step 1]
By using Equations
\eqref{eqn:quad_omega_A2_12}$_{(0,-1,0,0)}$,
\eqref{eqn:quad_omega_A2_13}$_{(0,0,-1,0)}$ and
\eqref{eqn:quad_omega_A2_14}$_{(0,0,0,-1)}$,
the functions on $\bmv_1-\bmv_2$, $\bmv_1-\bmv_3$ and $\bmv_1-\bmv_4$
can be calculated, respectively.
\item[Step 2]
By using Equations
\eqref{eqn:quad_omega_A2_13}$_{(1,1,0,0)}$,
\eqref{eqn:quad_omega_A2_12}$_{(1,0,1,0)}$ and
\eqref{eqn:quad_omega_A2_12}$_{(1,0,0,1)}$,
the functions on $2\bmv_1+\bmv_2$, $2\bmv_1+\bmv_3$ and $2\bmv_1+\bmv_4$
can be calculated, respectively.
\item[Step 3]
By using Equation \eqref{eqn:quad_omega_A2_23}$_{(2,0,0,0)}$, the function on $2\bmv_1$ can be calculated.
\end{description} 

In a similar manner, we can calculate all $\omega$-functions on $\bar{V}(\bml+\bmv_i)$, $i=1,\dots,4$, 
from those on $\bar{V}(\bml)$ for any $\bml\in\omega_{A_2+A_1}$.
Therefore we have completed the proof. 
\end{proof}

For later convenience, we here make the mention of $R_0$ briefly.
Its action on the parameters $b_i$ and $p$ is given by
\begin{equation}
 R_0:(b_0, b_1,b_2,b_3,p)\mapsto(b_1,pb_0,b_2,{b_3}^{-1},p),
\end{equation}
while that on the restricted $\omega$-functions,
which are on the following sublattice:
\begin{equation}
 \left.\left\{\sum_{i=1}^4l_i\bmv_i\,\right|\, l_1=l_2,~l_i\in\bbZ\right\}
 \cup
 \left.\left\{\sum_{i=1}^4l_i\bmv_i\,\right|\, l_1=l_2+1,~l_i\in\bbZ\right\}
 \subset\omega_{A_2+A_1},
\end{equation}
is given by
\begin{equation}
 R_0:\omega_{l_1,l_2,l_3,l_4}\mapsto
 \begin{cases}
 \omega_{l_1+1,l_2,l_3,l_4}&\text{if}\quad l_1=l_2,\\
 \omega_{l_1,l_2+1,l_3,l_4}&\text{if}\quad l_1=l_2+1.
 \end{cases}
\end{equation}

\subsection{Discrete Painlev\'e equations}\label{subsection:A2A1_DP}
In this section we consider the particular $f$-variables $f_1^{(j)}$, $j=1,2,3$, given by \eqref{eqn:def_f},
which can be expressed by the ratios of the $\omega$-functions $\omega_{l_1,l_2,l_3,l_4}$ as follows:
\begin{equation}\label{eqn:def_f_A2A1}
 f_1^{(1)}=\cfrac{\omega_2^{(1)}}{\omega_1^{(1)}}=\cfrac{\omega_{2,1,1,0}}{\omega_{1,1,0,0}},\quad
 f_1^{(2)}=\cfrac{\omega_2^{(2)}}{\omega_1^{(2)}}=\cfrac{\omega_1^{(3)}}{\omega_3^{(3)}}=\cfrac{\omega_{1,0,0,0}}{\omega_{1,1,0,1}},\quad
 f_1^{(3)}=\cfrac{\omega_2^{(3)}}{\omega_1^{(3)}}=\cfrac{\omega_{1,1,1,0}}{\omega_{1,0,0,0}}.
\end{equation}
These $f$-variables satisfy the relation \eqref{eqn:relation_f_A2A1},
which follows from the relations \eqref{eqns:A4_conditions_f}.
The action of $\widetilde{W}((A_2\rtimes A_1)^{(1)})$ on the three $f$-variables is given by
{\allowdisplaybreaks
\begin{subequations}
\begin{align}
 &w_0(f_1^{(3)})=\cfrac{p^3{b_0}^2 f_1^{(3)}(p{b_0}^2 b_3+{b_1}^2 f_1^{(1)})}{b_3+p{b_1}^2 f_1^{(1)}},\quad
 w_1(f_1^{(1)})=\cfrac{{b_1}^2 f_1^{(1)}({b_1}^2+p {b_0}^2b_3 f_1^{(3)})}{1+ p{b_0}^2b_3 f_1^{(3)}},\\
 &w_1(f_1^{(2)})=\cfrac{b_1}{p b_0 {b_3}^{3/2}f_1^{(3)}}
 \left(\cfrac{b_1r_0(f_1^{(3)})(1+p{b_0}^2b_3f_1^{(3)})}{pb_0{b_3}^{1/2}({b_1}^2+p{b_0}^2b_3f_1^{(3)})}-b_2\right),\\
 &w_2(f_1^{(1)})=\cfrac{p{b_3}^2f_1^{(3)}\left(p{b_1}^2+b_3 \pi(f_1^{(2)})\right)}{p {b_1}^2 b_3+r_0(f_1^{(3)})},\quad
 w_2(f_1^{(2)})=\cfrac{{b_1}^2f_1^{(2)}\left(p {b_1}^2 b_3+r_0(f_1^{(3)})\right)}{{b_0}^2 \left(p {b_0}^2 b_3+r_0(f_1^{(3)})\right)},\\
 &w_2(f_1^{(3)})=\cfrac{{b_0}^2f_1^{(3)} \left(p {b_0}^2 b_3+r_0(f_1^{(3)})\right)}{{b_1}^2 \left(p {b_1}^2 b_3+r_0(f_1^{(3)})\right)},\quad
 r_0(f_1^{(1)})=f_1^{(2)},\quad
 r_0(f_1^{(2)})=f_1^{(1)},\\
 &r_0(f_1^{(3)})=\cfrac{b_0 {b_3}^{3/2} (-b_1 b_2+p^2{b_0}^3 {b_3}^{1/2} f_1^{(3)})}{{b_1}^2 f_1^{(1)}},\quad
 r_1(f_1^{(1)})=f_1^{(3)},\quad
 r_1(f_1^{(3)})=f_1^{(1)},\\
 &r_1(f_1^{(2)})
 =-\cfrac{b_2 {b_3}^{1/2}}{p^3{b_0}^3 b_1 f_1^{(1)} f_1^{(3)}}
 -\cfrac{{b_0}^3 b_2 {b_3}^{3/2}}{p^2{b_0}^4 b_1 f_1^{(1)}}
 +\cfrac{{b_1}^2}{p^2{b_0}^4 f_1^{(3)}},\\
 &\pi(f_1^{(1)})=r_0(f_1^{(3)}),\quad
 \pi(f_1^{(2)})=\cfrac{b_1 (-b_0 b_2 {b_3}^{1/2}+p{b_1}^3 f_1^{(1)})}{p{b_0}^2 {b_3}^2 f_1^{(3)}}.
\end{align}
\end{subequations}
Note that 
\begin{equation}
 r_0(f_1^{(3)})=\cfrac{\omega_1^{(1)}}{\omega_3^{(1)}}.
\end{equation}
Moreover, the time evolutions of the $q$-Painlev\'e equations shown in \S \ref{subsection:A4_discrete_Painleve}
can be expressed by the elements of $\widetilde{W}((A_2\rtimes A_1)^{(1)})$ as follows:
\begin{equation}\label{eqn:timeevolution_A2}
 T_0=\rho_1\rho_2,\quad
 T_{13}={\rho_4}^2,\quad
 R_0=\pi r_0 w_1,\quad
 R_{13}=\rho_4,
\end{equation}
where $\rho_i$ are defined by \eqref{eqn:def_rho}.
Therefore, the birational actions of $T_0$, $T_{13}$, $R_0$ and $R_{13}$ are given by \eqref{eqns:intro_action_para} and \eqref{eqn:intro_action_var}.
As mentioned in Remark \ref{remark:relation_qps}, these actions give $q$-Painlev\'e equations \eqref{eqns:intro_qPs}.
}

\section{Proofs of Theorems \ref{maintheorem_hypercubestructure} and \ref{maintheorem_Lax}}
\label{section:proof_theorem}
In this section, we consider the following system of the partial difference equations:
\begin{subequations}\label{eqns:system_u}
\begin{align}
 &\cfrac{u(\bml+\ep_1+\ep_2)}{u(\bml)}
 =-\cfrac{\al_{l_1} u(\bml+\ep_1)-\be_{l_2} u(\bml+\ep_2)}{\al_{l_1} u(\bml+\ep_2)-\be_{l_2} u(\bml+\ep_1)},
 \label{eqn:system_u12}\\
 &\cfrac{u(\bml+\ep_2+\ep_3)}{u(\bml)}
 =-\cfrac{\be_{l_2} u(\bml+\ep_2)-\ga_{l_3} u(\bml+\ep_3)}{\be_{l_2} u(\bml+\ep_3)-\ga_{l_3} u(\bml+\ep_2)},
 \label{eqn:system_u23}\\
 &\cfrac{u(\bml+\ep_3+\ep_1)}{u(\bml)}
 =-\cfrac{\ga_{l_3} u(\bml+\ep_3)-\al_{l_1} u(\bml+\ep_1)}{\ga_{l_3} u(\bml+\ep_1)-\al_{l_1} u(\bml+\ep_3)},
 \label{eqn:system_u13}\\
 &\cfrac{u(\bml+\ep_1+\ep_4)}{u(\bml)}+\cfrac{u(\bml+\ep_4)}{u(\bml+\ep_1)}=-\al_{l_1} K_{l_4},
 \label{eqn:system_u14}\\
 &\cfrac{u(\bml+\ep_2+\ep_4)}{u(\bml)}+\cfrac{u(\bml+\ep_4)}{u(\bml+\ep_2)}=-\be_{l_2} K_{l_4},
 \label{eqn:system_u24}\\
 &\cfrac{u(\bml+\ep_3+\ep_4)}{u(\bml)}+\cfrac{u(\bml+\ep_4)}{u(\bml+\ep_3)}=-\ga_{l_3} K_{l_4},
 \label{eqn:system_u34}
\end{align}
\end{subequations}
where $\bml=\sum_{i=1}^4 l_i \ep_i\in\bbZ^4$ and $\{\ep_1,\dots,\ep_4\}$ is a standard basis for $\bbR^4$.
Here, $u(\bml)$ is a function from $\bbZ^4$ to $\bbC$
and $\Set{\al_l}{l\in\bbZ}$, $\Set{\be_l}{l\in\bbZ}$, $\Set{\ga_l}{l\in\bbZ}$ and $\Set{K_l}{l\in\bbZ}$ are complex parameters.
This system is obtained by assigning the quad-equations of ABS type to the faces of each 4-dimensional hypercube (4-cube)
(see \cite{JNS:paper3} and references therein).
The Lax equations for System \eqref{eqns:system_u} are given by the following\cite{JNS:paper3}:
{\allowdisplaybreaks
\begin{subequations}\label{eqns:Psi_u}
\begin{align}
 &\Psi_{l_1+1,l_2,l_3,l_4}
 =\de^{(1)}
 \begin{pmatrix}
 \cfrac{\mu}{\al_{l_1}}&-u(\bml+\ep_1)\\
 \cfrac{1}{u(\bml)}&-\cfrac{\mu}{\al_{l_1}}\,\cfrac{u(\bml+\ep_1)}{u(\bml)}
 \end{pmatrix}.
 \Psi_{l_1,l_2,l_3,l_4},\label{eqn:Psi_u_1}\\
 &\Psi_{l_1,l_2+1,l_3,l_4}
 =\de^{(2)}
 \begin{pmatrix}
 \cfrac{\mu}{\be_{l_2}}&-u(\bml+\ep_2)\\
 \cfrac{1}{u(\bml)}&-\cfrac{\mu}{\be_{l_2}}\,\cfrac{u(\bml+\ep_2)}{u(\bml)}
 \end{pmatrix}.
 \Psi_{l_1,l_2,l_3,l_4},\label{eqn:Psi_u_2}\\
 &\Psi_{l_1,l_2,l_3+1,l_4}
 =\de^{(3)}
 \begin{pmatrix}
 \cfrac{\mu}{\ga_{l_3}}&-u(\bml+\ep_3)\\
 \cfrac{1}{u(\bml)}&-\cfrac{\mu}{\ga_{l_3}}\,\cfrac{u(\bml+\ep_3)}{u(\bml)}
 \end{pmatrix}.
 \Psi_{l_1,l_2,l_3,l_4},\label{eqn:Psi_u_3}\\
 &\Psi_{l_1,l_2,l_3,l_4+1}
 =\de^{(4)}
 \begin{pmatrix}
 -\mu K_{l_4}&-u(\bml+\ep_4)\\
 \cfrac{1}{u(\bml)}&0
 \end{pmatrix}.
 \Psi_{l_1,l_2,l_3,l_4},\label{eqn:Psi_u_4}
\end{align}
\end{subequations}
where  $\de^{(i)}$, $i=1,\dots,4$, are arbitrary constants and $\mu$ is a spectral parameter.
The pairs of Equations \eqref{eqns:Psi_u} give the Lax pairs of \PDE s \eqref{eqns:system_u} (see Table \ref{table:Laxpairs}).
}

\begin{table}[t]
\begin{tabular}{|c|c|}
\hline
\PDE &Lax pair\\ \hline
\eqref{eqn:system_u12}&\eqref{eqn:Psi_u_1}, \eqref{eqn:Psi_u_2}\\
\eqref{eqn:system_u23}&\eqref{eqn:Psi_u_2}, \eqref{eqn:Psi_u_3}\\
\eqref{eqn:system_u13}&\eqref{eqn:Psi_u_1}, \eqref{eqn:Psi_u_3}\\
\eqref{eqn:system_u14}&\eqref{eqn:Psi_u_1}, \eqref{eqn:Psi_u_4}\\
\eqref{eqn:system_u24}&\eqref{eqn:Psi_u_2}, \eqref{eqn:Psi_u_4}\\
\eqref{eqn:system_u34}&\eqref{eqn:Psi_u_3}, \eqref{eqn:Psi_u_4}\\ \hline
\end{tabular}\quad
\begin{tabular}{|c|c|}
\hline
\PDE &Lax pair\\ \hline
\eqref{eqn:system_U12}&\eqref{eqn:phi_U_1}, \eqref{eqn:phi_U_2}\\
\eqref{eqn:system_U23}&\eqref{eqn:phi_U_2}, \eqref{eqn:phi_U_3}\\
\eqref{eqn:system_U13}&\eqref{eqn:phi_U_1}, \eqref{eqn:phi_U_3}\\
\eqref{eqn:system_U14}&\eqref{eqn:phi_U_1}, \eqref{eqn:phi_U_4}\\
\eqref{eqn:system_U24}&\eqref{eqn:phi_U_2}, \eqref{eqn:phi_U_4}\\
\eqref{eqn:system_U34}&\eqref{eqn:phi_U_3}, \eqref{eqn:phi_U_4}\\ \hline
\end{tabular}\quad
\begin{tabular}{|c|c|}
\hline
\PDE &Lax pair\\ \hline
\eqref{eqn:quad_omega_A2_12}&\eqref{eqn:phi_omega_1}, \eqref{eqn:phi_omega_2}\\
\eqref{eqn:quad_omega_A2_23}&\eqref{eqn:phi_omega_2}, \eqref{eqn:phi_omega_3}\\
\eqref{eqn:quad_omega_A2_13}&\eqref{eqn:phi_omega_1}, \eqref{eqn:phi_omega_3}\\
\eqref{eqn:quad_omega_A2_14}&\eqref{eqn:phi_omega_1}, \eqref{eqn:phi_omega_4}\\
\eqref{eqn:quad_omega_A2_24}&\eqref{eqn:phi_omega_2}, \eqref{eqn:phi_omega_4}\\
\eqref{eqn:quad_omega_A2_34}&\eqref{eqn:phi_omega_3}, \eqref{eqn:phi_omega_4}\\ \hline
\end{tabular}
\caption{The correspondences between \PDE s and Lax pairs.}
\label{table:Laxpairs}
\end{table}

\subsection{Proof of Theorem \ref{maintheorem_hypercubestructure}}
\label{subsection:geometric_reduction}
In this section, we show that the lattice $\omega_{A_2+A_1}$ can be obtained from the integer lattice $\bbZ^4$ 
with the \PDE s \eqref{eqns:system_u} by a geometric reduction.

Let
\begin{equation}\label{eqn:rel_u_U}
 u(\bml)=\cfrac{{\la_{l_1+l_2+l_3+l_4}}^{(l_1+l_2+l_3-2l_4)/2}}{U(\bml)},
\end{equation}
where $\bml=\sum_{i=1}^4 l_i \ep_i\in\bbZ^4$.
Here, $\la_0$ is a non-zero complex parameter and
\begin{equation}\label{eqn:constraints}
 \la_l=
 \begin{cases}
 \la_0&\text{if}\quad l=2n,\\
 \cfrac{1}{\la_0}&\text{if}\quad l=2n+1.
 \end{cases}
\end{equation}
Then, System \eqref{eqns:system_u} can be rewritten as the following:
\begin{subequations}\label{eqns:system_U}
\begin{align}
 &\cfrac{U(\bml+\ep_1+\ep_2)}{U(\bml)}
 =-\la_{l_1+l_2+l_3+l_4}\,\cfrac{\al_{l_1} U(\bml+\ep_1)-\be_{l_2} U(\bml+\ep_2)}{\al_{l_1} U(\bml+\ep_2)-\be_{l_2} U(\bml+\ep_1)},
 \label{eqn:system_U12}\\
 &\cfrac{U(\bml+\ep_2+\ep_3)}{U(\bml)}
 =-\la_{l_1+l_2+l_3+l_4}\,\cfrac{\be_{l_2} U(\bml+\ep_2)-\ga_{l_3} U(\bml+\ep_3)}{\be_{l_2} U(\bml+\ep_3)-\ga_{l_3} U(\bml+\ep_2)},
 \label{eqn:system_U23}\\
 &\cfrac{U(\bml+\ep_3+\ep_1)}{U(\bml)}
 =-\la_{l_1+l_2+l_3+l_4}\,\cfrac{\ga_{l_3} U(\bml+\ep_3)-\al_{l_1} U(\bml+\ep_1)}{\ga_{l_3} U(\bml+\ep_1)-\al_{l_1} U(\bml+\ep_3)},
 \label{eqn:system_U13}\\
 &\cfrac{U(\bml)}{U(\bml+\ep_1+\ep_4)}+{\la_{l_1+l_2+l_3+l_4}}^2\,\cfrac{U(\bml+\ep_1)}{U(\bml+\ep_4)}
 =-\al_{l_1} K_{l_4}{\la_{l_1+l_2+l_3+l_4}}^{1/2},
 \label{eqn:system_U14}\\
 &\cfrac{U(\bml)}{U(\bml+\ep_2+\ep_4)}+{\la_{l_1+l_2+l_3+l_4}}^2\,\cfrac{U(\bml+\ep_2)}{U(\bml+\ep_4)}
 =-\be_{l_2} K_{l_4}{\la_{l_1+l_2+l_3+l_4}}^{1/2},
 \label{eqn:system_U24}\\
 &\cfrac{U(\bml)}{U(\bml+\ep_3+\ep_4)}+{\la_{l_1+l_2+l_3+l_4}}^2\,\cfrac{U(\bml+\ep_3)}{U(\bml+\ep_4)}
 =-\ga_{l_3} K_{l_4}{\la_{l_1+l_2+l_3+l_4}}^{1/2}.
 \label{eqn:system_U34}
\end{align}
\end{subequations}
Moreover, by imposing the following $(1,1,1,1)$-periodic condition:
\begin{equation}\label{eqn:1111periodic_U}
 U(\bml)=U(\bml+\ep_1+\ep_2+\ep_3+\ep_4),
\end{equation}
for $\bml\in\bbZ^4$, with the following condition of the parameters:
\begin{equation}\label{eqn:condition_albegaK}
 \al_l=p^{-l}\al_0,\quad
 \be_l=p^{-l}\be_0,\quad
 \ga_l=p^{-l}\ga_0,\quad
 K_l=p^l K_0,
\end{equation}
where $p$ is a non-zero complex parameter,
System \eqref{eqns:system_U} becomes the system of $q$-difference equations (in this case the shift parameter is given by $p$).

We define the transformations $\hrho_i$, $i=1,\dots,4$, by the following actions:
\begin{subequations}\label{eqns:hrho_uabgK}
\begin{align}
 &\hrho_1:(U(\bml),\al_0,\be_0,\ga_0,K_0,\la_0,p)\mapsto(U(\bml+\ep_1),p^{-1}\al_0,\be_0,\ga_0,K_0,{\la_0}^{-1},p),\\
 &\hrho_2:(U(\bml),\al_0,\be_0,\ga_0,K_0,\la_0,p)\mapsto(U(\bml+\ep_2),\al_0,p^{-1}\be_0,\ga_0,K_0,{\la_0}^{-1},p),\\
 &\hrho_3:(U(\bml),\al_0,\be_0,\ga_0,K_0,\la_0,p)\mapsto(U(\bml+\ep_3),\al_0,\be_0,p^{-1}\ga_0,K_0,{\la_0}^{-1},p),\\
 &\hrho_4:(U(\bml),\al_0,\be_0,\ga_0,K_0,\la_0,p)\mapsto(U(\bml+\ep_4),\al_0,\be_0,\ga_0,pK_0,{\la_0}^{-1},p),
\end{align}
\end{subequations}
which imply that $\hrho_i$ is a shift operator of $\ep_i$-direction on $\bbZ^4$.
In addition, we also introduce a transformation $\hR_0$ as follows.
Its action on the parameters is defined by
\begin{equation}
 \hR_0:(\al_0,\be_0,\ga_0,K_0,\la_0,p)\mapsto(\be_0,p^{-1}\al_0,\ga_0,K_0,{\la_0}^{-1},p),
\end{equation}
while that on the function $U(\bml)$ is defined by
\begin{equation}
 \hR_0(U(\bml))
 =\begin{cases}
 U(\bml+\ep_1)&\text{if}\quad \bml\in\calr^{(1)},\\
 U(\bml+\ep_2)&\text{if}\quad \bml\in\calr^{(2)},
 \end{cases}
\end{equation}
where
\begin{equation}
 \calr^{(1)}=\set{\sum_{i=1}^4 l_i \ep_i}{l_i\in\bbZ,~ l_1=l_2},\quad
 \calr^{(2)}=\set{\sum_{i=1}^4 l_i \ep_i}{l_i\in\bbZ,~ l_1=l_2+1}.
\end{equation}
These imply that $\hR_0$ is a zigzag-shift operator on the sublattice 
\begin{equation}
 \calR=\calr^{(1)}\cup\calr^{(2)}\subset\bbZ^4,
\end{equation}
that is,
\begin{equation}\label{eqns:hR0_onR}
 \hR_0(\bml)
 =\begin{cases}
 \hrho_1(\bml)&\text{if}\quad \bml\in\calr^{(1)},\\
 \hrho_2(\bml)&\text{if}\quad \bml\in\calr^{(2)},
 \end{cases}\qquad
 \hR_0^{~-1}(\bml)
 =\begin{cases}
 {\hrho_2}^{~-1}(\bml)&\text{if}\quad \bml\in\calr^{(1)},\\
 {\hrho_1}^{~-1}(\bml)&\text{if}\quad \bml\in\calr^{(2)}.
 \end{cases}
\end{equation}
In general, for a function $F=F(U(\bml),\al_0,\be_0,\ga_0,K_0,\la_0,p)$, 
we let a transformation $w\in\langle\hrho_1,\dots,\hrho_4,\hR_0\rangle$ act as 
\begin{equation}
 w(F)=F\Big(w(U(\bml)),w(\al_0),w(\be_0),w(\ga_0),w(K_0),w(\la_0),w(p)\Big),
\end{equation}
that is, the transformation $w$ act on the arguments from the left.

Finally, letting
\begin{equation}\label{eqn:rel_U_omega}
 \omega_{l_1,l_2,l_3,l_4}=H_{l_1,l_2,l_3,l_4}U(\bml),\quad
 b_0=\cfrac{\ga_0}{\al_0},\quad
 b_1=\cfrac{\ga_0}{\be_0},\quad
 b_2=\ga_0K_0,
\end{equation}
where $\bml=\sum_{i=1}^4 l_i \ep_i$,
we obtain the fundamental relations on the lattice $\omega_{A_2+A_1}$ \eqref{eqns:quad_omega_A2} from System \eqref{eqns:system_U}.
Here, the gauge factor $H_{l_1,l_2,l_3,l_4}$ is defined by
\begin{align}
 H_{l_1,l_2,l_3,l_4}
 =&\,\ii^{(\log{\al_{l_1}\be_{l_2}\ga_{l_3}{K_{l_4}}^3})/\log{p}}
 \ee^{\left(
 (\log p^3 {\al_{l_1}}^2 {\be_{l_2}}^2{\ga_{l_3}}^{-4})^2
 +12 (\log \al_{l_1}{\be_{l_2}}^{-1})^2
 \right)/16 \log{p}}\notag\\
 &\times \left(\cfrac{\ga_{l_3}}{{\al_{l_1}}^{1/2}{\be_{l_2}}^{1/2}}\right)^{3/2},
\end{align}
where $\ii=\sqrt{-1}$.
Furthermore, the actions of transformations $\hrho_i$, $i=1,\dots,4$, and $\hR_0$
correspond to those of $\rho_i$, $i=1,\dots,4$, and $R_0$ which are elements of $\widetilde{W}((A_2\rtimes A_1)^{(1)})$, respectively.
We note here that the reduction from System \eqref{eqns:system_u} to System \eqref{eqns:quad_omega_A2}
causes the reduction of the underlying lattice (see Figure \ref{fig:4cube_barV}): 
\begin{equation*}
 \bbZ^4\to \bbZ^4/\bbZ(\ep_1+\ep_2+\ep_3+\ep_4)\cong \omega_{A_2+A_1}.
\end{equation*}
The reduction from $\bbZ^4$ with System \eqref{eqns:system_u} to the lattice $\omega_{A_2+A_1}$ is referred to as geometric reduction\cite{JNS2014:MR3291391}
and then the lattice $\omega_{A_2+A_1}$ is said to have the reduced hypercube structure.
Therefore, we have completed the proof of Theorem \ref{maintheorem_hypercubestructure}.

\begin{figure}[t]
\begin{center}
\includegraphics[width=0.8\textwidth]{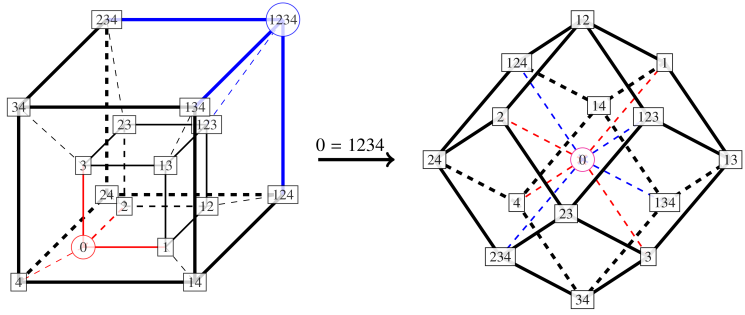}
\end{center}
\caption{A $(1,1,1,1)$-reduction from a 4-cube to a rhombic dodecahedron with a center.}
\label{fig:4cube_barV}
\end{figure}

\subsection{Proof of Theorem \ref{maintheorem_Lax}}
\label{subsection:Laxpairs}
{\allowdisplaybreaks
In this section, we construct the Lax pairs of the $q$-Painlev\'e equations \eqref{eqns:intro_qPs}
from the Lax equations \eqref{eqns:Psi_u} by using the reduction given in \S \ref{subsection:geometric_reduction}.

By the gauge transformations \eqref{eqn:rel_u_U} and
\begin{equation}
 \Psi_{l_1,l_2,l_3,l_4}
 =\ii^{-l_1-l_2-l_3+3l_4}
 \begin{pmatrix}
 U(\bml)^{-1}&0\\0&\ii\,{\la_{l_1+l_2+l_3+l_4}}^{(l_1+l_2+l_3-2l_4)/2}
 \end{pmatrix}.
 \phi_{l_1,l_2,l_3,l_4},
\end{equation}
the Lax equations \eqref{eqns:Psi_u} can be rewritten as the following:
{\allowdisplaybreaks
\begin{subequations}\label{eqns:phi_U}
\begin{align}
  &\phi_{l_1+1,l_2,l_3,l_4}
 =\de^{(1)}
 \begin{pmatrix}
 \cfrac{\ii\,\mu}{\al_{l_1}}\,\cfrac{U(\bml+\ep_1)}{U(\bml)}&\cfrac{1}{{\la_{l_1+l_2+l_3+l_4}}^{1/2}}\\
 {\la_{l_1+l_2+l_3+l_4}}^{1/2}&-\cfrac{\ii\,\mu}{\al_{l_1}}\,\cfrac{U(\bml)}{U(\bml+\ep_1)}
 \end{pmatrix}.
 \phi_{l_1,l_2,l_3,l_4},\label{eqn:phi_U_1}\\
 &\phi_{l_1,l_2+1,l_3,l_4}
 =\de^{(2)}
 \begin{pmatrix}
 \cfrac{\ii\,\mu}{\be_{l_2}}\,\cfrac{U(\bml+\ep_2)}{U(\bml)}&\cfrac{1}{{\la_{l_1+l_2+l_3+l_4}}^{1/2}}\\
 {\la_{l_1+l_2+l_3+l_4}}^{1/2}&-\cfrac{\ii\,\mu}{\be_{l_2}}\,\cfrac{U(\bml)}{U(\bml+\ep_2)}
 \end{pmatrix}.
 \phi_{l_1,l_2,l_3,l_4},\label{eqn:phi_U_2}\\
 &\phi_{l_1,l_2,l_3+1,l_4}
 =\de^{(3)}
 \begin{pmatrix}
 \cfrac{\ii\,\mu}{\ga_{l_3}}\,\cfrac{U(\bml+\ep_3)}{U(\bml)}&\cfrac{1}{{\la_{l_1+l_2+l_3+l_4}}^{1/2}}\\
 {\la_{l_1+l_2+l_3+l_4}}^{1/2}&-\cfrac{\ii\,\mu}{\ga_{l_3}}\,\cfrac{U(\bml)}{U(\bml+\ep_3)}
 \end{pmatrix}.
 \phi_{l_1,l_2,l_3,l_4},\label{eqn:phi_U_3}\\
 &\phi_{l_1,l_2,l_3,l_4+1}
 =\de^{(4)}
 \begin{pmatrix}
 -\ii\,\mu K_{l_4}\,\cfrac{U(\bml+\ep_4)}{U(\bml)}&\la_{l_1+l_2+l_3+l_4}\\
 \cfrac{1}{\la_{l_1+l_2+l_3+l_4}}&0
 \end{pmatrix}.
 \phi_{l_1,l_2,l_3,l_4}.\label{eqn:phi_U_4}
\end{align}
\end{subequations}
These give the Lax pairs of \PDE s \eqref{eqns:system_U} (see Table \ref{table:Laxpairs}).
Moreover, by the reduction \eqref{eqn:1111periodic_U} with \eqref{eqn:condition_albegaK}
and the replacement \eqref{eqn:rel_U_omega},
the Lax equations \eqref{eqns:phi_U} can be rewritten as the following:
\begin{subequations}\label{eqns:phi_omega}
\begin{align}
  &\phi_{l_1+1,l_2,l_3,l_4}
 =\de^{(1)}
 \begin{pmatrix}
 -p^{-l_1+l_2+l_3-1}\cfrac{b_1}{b_0}\,\cfrac{\omega_{l_1+1,l_2,l_3,l_4}}{\omega_{l_1,l_2,l_3,l_4}}\,x
 &\cfrac{1}{{\la_{l_1+l_2+l_3+l_4}}^{1/2}}\\
 \rule{0pt}{2em}
 {\la_{l_1+l_2+l_3+l_4}}^{1/2}
 &-p^{3l_1-l_2-l_3+1}\cfrac{{b_0}^3}{b_1}\,\cfrac{\omega_{l_1,l_2,l_3,l_4}}{\omega_{l_1+1,l_2,l_3,l_4}}\,x
 \end{pmatrix}.
 \phi_{l_1,l_2,l_3,l_4},\label{eqn:phi_omega_1}\\
 &\phi_{l_1,l_2+1,l_3,l_4}
 =\de^{(2)}
 \begin{pmatrix}
 -p^{l_1-l_2+l_3-1}\cfrac{b_0}{b_1}\,\cfrac{\omega_{l_1,l_2+1,l_3,l_4}}{\omega_{l_1,l_2,l_3,l_4}}\,x
 &\cfrac{1}{{\la_{l_1+l_2+l_3+l_4}}^{1/2}}\\
 \rule{0pt}{2em}
 {\la_{l_1+l_2+l_3+l_4}}^{1/2}
 &-p^{-l_1+3l_2-l_3+1}\cfrac{{b_1}^3}{b_0}\,\cfrac{\omega_{l_1,l_2,l_3,l_4}}{\omega_{l_1,l_2+1,l_3,l_4}}\,x
 \end{pmatrix}.
 \phi_{l_1,l_2,l_3,l_4},\label{eqn:phi_omega_2}\\
 &\phi_{l_1,l_2,l_3+1,l_4}
 =\de^{(3)}
 \begin{pmatrix}
 -p^{l_1+l_2-l_3-1}b_0b_1\cfrac{\omega_{l_1,l_2,l_3+1,l_4}}{\omega_{l_1,l_2,l_3,l_4}}\,x
 &\cfrac{1}{{\la_{l_1+l_2+l_3+l_4}}^{1/2}}\\
 \rule{0pt}{2em}
 {\la_{l_1+l_2+l_3+l_4}}^{1/2}
 &-p^{-l_1-l_2+3l_3+1}\cfrac{1}{b_0b_1}\,\cfrac{\omega_{l_1,l_2,l_3,l_4}}{\omega_{l_1,l_2,l_3+1,l_4}}\,x
 \end{pmatrix}.
 \phi_{l_1,l_2,l_3,l_4},\label{eqn:phi_omega_3}\\
 &\phi_{l_1,l_2,l_3,l_4+1}
 =\de^{(4)}
 \begin{pmatrix}
 p^{l_4}b_2\cfrac{\omega_{l_1,l_2,l_3,l_4+1}}{\omega_{l_1,l_2,l_3,l_4}}\,x&\la_{l_1+l_2+l_3+l_4}\\
 \rule{0pt}{2em}
 \cfrac{1}{\la_{l_1+l_2+l_3+l_4}}&0
 \end{pmatrix}.
 \phi_{l_1,l_2,l_3,l_4},\label{eqn:phi_omega_4}
\end{align}
\end{subequations}
where 
\begin{equation}
 x=\cfrac{\mu}{\ga_0}.
\end{equation}
These give the Lax pairs of \PDE s \eqref{eqns:quad_omega_A2} (see Table \ref{table:Laxpairs}).
}

Now we are in a position to construct the Lax pairs of the $q$-Painlev\'e equations. 
We first lift the action of $\langle \hrho_1,\dots,\hrho_4,\hR_0\rangle$ up to the Lax equations \eqref{eqns:phi_omega} by
\begin{subequations}
\begin{align}
 \hrho_1&:(\de^{(1)},\de^{(2)},\de^{(3)},\de^{(4)},\mu,\phi_{l_1,l_2,l_3,l_4})
 \mapsto(\de^{(1)},\de^{(2)},\de^{(3)},\de^{(4)},\mu,\phi_{l_1+1,l_2,l_3,l_4}),\\
 \hrho_2&:(\de^{(1)},\de^{(2)},\de^{(3)},\de^{(4)},\mu,\phi_{l_1,l_2,l_3,l_4})
 \mapsto(\de^{(1)},\de^{(2)},\de^{(3)},\de^{(4)},\mu,\phi_{l_1,l_2+1,l_3,l_4}),\\
 \hrho_3&:(\de^{(1)},\de^{(2)},\de^{(3)},\de^{(4)},\mu,\phi_{l_1,l_2,l_3,l_4})
 \mapsto(\de^{(1)},\de^{(2)},\de^{(3)},\de^{(4)},\mu,\phi_{l_1,l_2,l_3+1,l_4}),\\
 \hrho_4&:(\de^{(1)},\de^{(2)},\de^{(3)},\de^{(4)},\mu,\phi_{l_1,l_2,l_3,l_4})
 \mapsto(\de^{(1)},\de^{(2)},\de^{(3)},\de^{(4)},\mu,\phi_{l_1,l_2,l_3,l_4+1}),\\
 \hR_0&:(\de^{(1)},\de^{(2)},\de^{(3)},\de^{(4)},\mu,\phi_{l_1,l_2,l_3,l_4})
 \mapsto(\de^{(2)},\de^{(1)},\de^{(3)},\de^{(4)},\mu,\phi_{l_1,l_2,l_3,l_4}^R),
\end{align}
\end{subequations}
where
\begin{equation}
 \phi_{l_1,l_2,l_3,l_4}^R=
 \begin{cases}
  \phi_{l_1+1,l_2,l_3,l_4}&\text{if}\quad l_1=l_2,\\
  \phi_{l_1,l_2+1,l_3,l_4}&\text{if}\quad l_1=l_2+1.
 \end{cases}
\end{equation}
By letting
\begin{equation}
 \phi_{0,0,0,-1}=\begin{pmatrix}\omega_{0,0,0,-1}&0\\0&\omega_{0,0,0,0}\end{pmatrix}.\Phi,
\end{equation}
the action of $\langle\hrho_1\dots,\hrho_4,\hR_0\rangle$ on $\Phi$ is given by
{\allowdisplaybreaks
\begin{subequations}
\begin{align}
 &\hrho_1(\Phi)
 =\de_1\begin{pmatrix}
 -\cfrac{b_1}{b_0 p}\,x
 &\cfrac{{b_1}^2}{b_0 b_3 (p^2{b_0}^3 {b_3}^{1/2} f_1^{(3)}-b_1 b_2)}\\
 \rule{0pt}{2em}
 \cfrac{f_1^{(3)}}{{b_3}^{1/2}}
 &-\cfrac{p{b_0}^2 b_1 f_1^{(3)}}{{b_3}^{3/2} (p^2{b_0}^3 {b_3}^{1/2} f_1^{(3)}-b_1 b_2)}\,x
 \end{pmatrix}.\Phi,\\
 &\hrho_2(\Phi)
 =\de_2\begin{pmatrix}
 -\cfrac{b_0}{p b_1}\,x
 &\cfrac{{b_0}^2}{b_1 b_3 \left(p^2{b_1}^3 {b_3}^{1/2} {\hrho_1}^{~-1}(f_1^{(1)})-b_0 b_2\right)}\\
 \rule{0pt}{2em}
 \cfrac{{\hrho_1}^{~-1}(f_1^{(1)})}{{b_3}^{1/2}} 
 &-\cfrac{pb_0 {b_1}^2 {\hrho_1}^{~-1}(f_1^{(1)})}{{b_3}^{3/2} \left(p^2{b_1}^3 {b_3}^{1/2} {\hrho_1}^{~-1}(f_1^{(1)})-b_0 b_2\right)}\,x
 \end{pmatrix}.\Phi,\\
 &\hrho_3(\Phi)
 =\de_3\begin{pmatrix}
 -\cfrac{b_0 b_1}{p}\,x
 &\cfrac{{b_3}^{1/2}}{\hrho_3\left(f_1^{(2)} f_1^{(3)}\right)}\\
  \rule{0pt}{2.2em}
 \cfrac{b_0 b_1 \left(b_0 b_1 \hrho_3\left(f_1^{(2)} f_1^{(3)}\right)+b_2 {b_3}^{3/2}\right)}{p^2{b_3}^{5/2}}
 &-\cfrac{b_0 b_1 \hrho_3\left(f_1^{(2)} f_1^{(3)}\right)+b_2 {b_3}^{3/2}}{p {b_3}^2 \hrho_3\left(f_1^{(2)} f_1^{(3)}\right)}\,x
 \end{pmatrix}.\Phi,\\
 &\hrho_4(\Phi)
 =\de_4\begin{pmatrix}
 \cfrac{b_2}{p}\,x
 &\cfrac{1}{b_3}\\
  \rule{0pt}{2.2em}
 \cfrac{p b_0 f_1^{(3)} \left(p{b_0}^3 \hrho_4(f_1^{(3)})-b_1 b_2 {b_3}^{1/2}\right)}{{b_1}^2 b_3}
 &0
 \end{pmatrix}.\Phi,\\
 &\hR_0(\Phi)=\hrho_1(\Phi),
\end{align}
\end{subequations}
where $f_1^{(j)}$, $j=1,2,3$, are given by \eqref{eqn:def_f_A2A1} and satisfy the relation \eqref{eqn:relation_f_A2A1}.
}
Next, let us define
\begin{equation}
 \spT=\hrho_1\hrho_2\hrho_3\hrho_4,\quad
 \hT_0=\hrho_1\hrho_2,\quad
 \hT_{13}={\hrho_4}^2,\quad
 \hR_{13}=\hrho_4.
\end{equation}

\begin{remark}
Under the actions on the $f$-variables $f_1^{(j)}$, $j=1,2,3$, and the parameters $b_i$, $i=0,\dots,3$, and $p$,
the transformations $\hT_0$, $\hT_{13}$, $\hR_0$ and $\hR_{13}$ are respectively equivalent to 
the transformations $T_0$, $T_{13}$, $R_0$ and $R_{13}$, which are elements of $\widetilde{W}((A_2\rtimes A_1)^{(1)})$,
and the spectral operator $\spT$ can be regarded as an identity mapping.
\end{remark}

The actions of $\spT$, $\hT_0$, $\hT_{13}$, $\hR_0$ and $\hR_{13}$
on the spectral parameter $x$ are given by
\begin{equation}
 \spT(x)=px,\quad
 \hT_0(x)=\hT_{13}(x)=\hR_0(x)=\hR_{13}(x)=x,
\end{equation}
while those on the wave function $\Phi$ are given by the following:
\begin{subequations}
\begin{align}
 &\spT(\Phi)=\de_1\de_2\de_3\de_4\, A.\Phi,\quad
 \hT_0(\Phi)=\de_1\de_2\, B_{T0}.\Phi,\quad
 \hT_{13}(\Phi)={\de_4}^2\, B_{T13}.\Phi,\\
 &\hR_0(\Phi)=\de_1\, B_{R0}.\Phi,\quad
 \hR_{13}(\Phi)=\de_4\, B_{R13}.\Phi,
\end{align}
\end{subequations}
where
\begin{align}\label{eqn:intro_linear_martrix}
 A
 =&\begin{pmatrix}
 \cfrac{b_2}{p}\,x 
 &b_3\\
 \rule{0pt}{2em}
 -\cfrac{{b_3}^{1/2} (p^2 {b_0}^3 {b_3}^{1/2}f_1^{(3)}-b_1 b_2) (b_0 b_2 {b_3}^{1/2}+p {b_0}^3 b_2{b_3}^{3/2}f_1^{(3)}-p {b_1}^3f_1^{(1)})}{p^3 {b_0}^3 {b_1}^3f_1^{(1)}f_1^{(3)}} 
 &0 
 \end{pmatrix}\notag\\
 &.\begin{pmatrix}
 -p b_0 b_1 x 
 &-\cfrac{p^3 {b_0}^4 b_1 {b_3}^{1/2}f_1^{(3)}}{b_0 b_2 {b_3}^{1/2}+p {b_0}^3 b_2{b_3}^{3/2}f_1^{(3)}-p {b_1}^3f_1^{(1)}}\\
 \rule{0pt}{2em}
 \cfrac{b_1 (p {b_1}^3f_1^{(1)}-b_0 b_2 {b_3}^{1/2})}{p {b_0}^2 {b_3}^{5/2}f_1^{(3)}} 
 &\cfrac{p b_0 b_1(p {b_1}^3f_1^{(1)}-b_0 b_2 {b_3}^{1/2})}{b_0 b_2 {b_3}^{5/2}+p {b_0}^3 b_2{b_3}^{7/2}f_1^{(3)}-p {b_1}^3{b_3}^2f_1^{(1)}} \,x
 \end{pmatrix}\notag\\
 &.\begin{pmatrix}
 -\cfrac{b_0}{b_1}\,x &
 \cfrac{p {b_0}^2 {b_3}^{3/2}}{b_1(p {b_1}^3f_1^{(1)}-b_0 b_2 {b_3}^{1/2})}\\
 \rule{0pt}{2em}
 {b_3}^{1/2}f_1^{(1)} 
 &\cfrac{p b_0 {b_1}^2 {b_3}^2f_1^{(1)}}{b_0 b_2 {b_3}^{1/2}-p {b_1}^3f_1^{(1)}}\,x
 \end{pmatrix}\notag\\
 &.\begin{pmatrix}
 -\cfrac{b_1}{p b_0}\,x 
 &\cfrac{{b_1}^2}{b_0 b_3 (p^2 {b_0}^3 {b_3}^{1/2}f_1^{(3)}-b_1 b_2)}\\
 \rule{0pt}{2em}
 \cfrac{f_1^{(3)}}{{b_3}^{1/2}} 
 &-\cfrac{p {b_0}^2 b_1f_1^{(3)}}{{b_3}^{3/2} (p^2 {b_0}^3 {b_3}^{1/2}f_1^{(3)}-b_1 b_2)}\,x
 \end{pmatrix},
\end{align}
{\allowdisplaybreaks
\begin{subequations}\label{eqns:intro_matrices_B}
\begin{align}
 B_{T0}
 =&\begin{pmatrix}
 -\cfrac{b_0}{b_1}\,x 
 &\cfrac{p {b_0}^2 {b_3}^{3/2}}{b_1(p {b_1}^3f_1^{(1)}-b_0 b_2 {b_3}^{1/2})}\\
 \rule{0pt}{2em}
 {b_3}^{1/2}f_1^{(1)} 
 &\cfrac{p b_0 {b_1}^2 {b_3}^2f_1^{(1)}}{b_0 b_2 {b_3}^{1/2}-p {b_1}^3f_1^{(1)}}\,x
 \end{pmatrix}\notag\\
 &.\begin{pmatrix}
 -\cfrac{b_1}{p b_0}\,x 
 &\cfrac{{b_1}^2}{b_0 b_3 (p^2 {b_0}^3 {b_3}^{1/2}f_1^{(3)}-b_1 b_2)}\\
 \rule{0pt}{2em}
 \cfrac{f_1^{(3)}}{{b_3}^{1/2}} 
 &-\cfrac{p {b_0}^2 b_1f_1^{(3)}}{{b_3}^{3/2} (p^2 {b_0}^3 {b_3}^{1/2}f_1^{(3)}-b_1 b_2)}\,x
 \end{pmatrix},
 \label{eqn:intro_T0_Phi}\\
 B_{T13}
 =&\begin{pmatrix}
 b_2 x 
 &b_3\\
 \rule{0pt}{2em}
 \cfrac{p b_1 {b_3}^{1/2}f_1^{(1)} ({b_1}^3 {b_3}^{1/2}f_1^{(2)}-b_0 b_2) \left(p {b_0}^3 {b_3}^{1/2} T_{13}(f_1^{(3)})-p b_1 b_2\right)}{{b_0}^2 (p^2 {b_0}^3 {b_3}^{1/2}f_1^{(3)}-b_1 b_2)} 
 &0
 \end{pmatrix}\notag\\
 &.\begin{pmatrix}
 \cfrac{b_2}{p}\,x 
 &\cfrac{1}{b_3}\\
 \rule{0pt}{2em}
 \cfrac{p b_0f_1^{(3)}}{b_1 b_3}
  \left(\cfrac{p {b_1}^2f_1^{(1)} (b_0 b_2-{b_1}^3 {b_3}^{1/2}f_1^{(2)})}{b_1 b_2-p^2 {b_0}^3 {b_3}^{1/2}f_1^{(3)}}-b_2 {b_3}^{1/2}\right)
 &0 
 \end{pmatrix},
 \label{eqn:intro_T13_Phi}\\
 B_{R0}
 =&\begin{pmatrix}
 -\cfrac{b_1}{pb_0}\,x
 &\cfrac{{b_1}^2}{b_0 b_3 (p^2{b_0}^3 {b_3}^{1/2} f_1^{(3)}-b_1 b_2)}\\
 \rule{0pt}{2em}
 \cfrac{f_1^{(3)}}{{b_3}^{1/2}}
 &-\cfrac{p{b_0}^2 b_1 f_1^{(3)}}{{b_3}^{3/2} (p^2{b_0}^3 {b_3}^{1/2} f_1^{(3)}-b_1 b_2)}\,x
 \end{pmatrix},
 \label{eqn:intro_R0_Phi}\\
 B_{R13}
 =&\begin{pmatrix}
 \cfrac{b_2}{p}\,x
 &\cfrac{1}{b_3}\\
  \rule{0pt}{2.2em}
 \cfrac{p b_0 f_1^{(3)} \left(p{b_0}^3 R_{13}(f_1^{(3)})-b_1 b_2 {b_3}^{1/2}\right)}{{b_1}^2 b_3}
 &0
 \end{pmatrix}.
 \label{eqn:intro_R13_Phi}
\end{align}
\end{subequations}
Therefore, we finally obtain Theorem \ref{maintheorem_Lax} by the following correspondence:
}
\begin{subequations}
\begin{align}
 &\hT_0=T_0,\quad
 \hT_{13}=T_{13},\quad
 \hR_0=R_0,\quad
 \hR_{13}=R_{13},\\
 &\de_1=\de_2=\de_3=\de_4=1.
\end{align}
\end{subequations}
}
\section{Concluding remarks}
\label{ConcludingRemarks}
In this paper, we constructed the $\omega$-lattice of type $A_4^{(1)}$.
The $\omega$-lattice provides the informations about 
how a system of partial difference equations can be reduced to $A_4^{(1)}$-surface $q$-Painlev\'e equations. 
We will show how to use this information in forthcoming paper (N. Joshi, N. Nakazono and Y. Shi, in preparation).
We also constructed another important lattice $\omega_{A_2+A_1}$
and showed that it has the reduced hypercube structure.
Moreover, by using this structure, we constructed the Lax pairs of the $q$-Painlev\'e equations \eqref{eqns:intro_qPs}.
The distinguishing feature of the Lax pairs given in this paper 
as compared with those in the other works, e.g. \cite{MurataM2009:MR2485835,JS1996:MR1403067}, is that
their coefficient matrices can be factorized into the products of matrices which are of degree one in the spectral parameter $x$.
This property enables us to construct the Lax pairs of symmetric discrete Painlev\'e equations, e.g. 
$q$-P$_{\rm III}({D_7^{(1)}})$ \eqref{eqn:intro_qP3_1} and $q$-P$_{\rm IV}$ \eqref{eqn:intro_qP4_1},
which can be obtained by projective reductions \cite{KN2015:MR3340349,KNT2011:MR2773334}.
\subsection*{Acknowledgment}
This research was supported by an Australian Laureate Fellowship \# FL120100094 and grant \# DP130100967 from the Australian Research Council.
\appendix
\section{Proof of Lemma \ref{lemma:tau_A4}}
\label{section:proof_tauA4}
In this section, we define the transformation group $\widetilde{W}(A_4^{(1)})$ with its linear action
and show it forms the extended affine Weyl group of type $A_4^{(1)}$.
Moreover, we lift its action to the birational action on the parameters and the $\tau$-variables.

First, we define the transformation group $\widetilde{W}(A_4^{(1)})=\langle s_0,s_1,s_2,s_3,s_4,\sigma,\iota\rangle$.
Let $(f,g)$ be inhomogeneous coordinate of $\mathbb{P}^1\times\mathbb{P}^1$.
We consider the following eight base points of $\mathbb{P}^1\times \mathbb{P}^1$:
\begin{subequations}\label{eqns:basepoints}
\begin{align}
 &p_1:(f,g)=(-{a_0}^{-1}a_3,0),
 &&p_2:(f,g)=(-{a_0}^{-1}{a_1}^{-1}a_3,0),\\
 &p_3:(f,g)=(-{a_0}^{-1}a_2a_3,\infty),
 &&p_4:(f,g)=(0,-a_0{a_2}^{-1}),\\
 &p_5:(f,g)=(0,-a_0{a_2}^{-1}a_4),
 &&p_6:(f,g)=(\infty,-a_0{a_2}^{-1}{a_3}^{-1}),\\
 &p_\infty:(f,g)=(\infty,\infty),
 &&p_7:(f,g;f/g)=(\infty,\infty;-{a_0}^{-1}a_2a_3),
\end{align}
\end{subequations}
where $a_i$, $i=0,\dots,4$, are non-zero complex parameters.
Let $\ep: X \to \mathbb{P}^1\times\mathbb{P}^1$ denote blow up of $\mathbb{P}^1\times\mathbb{P}^1$ at the points \eqref{eqns:basepoints}.
The linear equivalence classes of the total transform of the coordinate lines $f$=constant and $g$=constant are denoted by $h_1$ and $h_2$, respectively. 
The Picard group of $X$, denoted by Pic$(X)$, is given by
\begin{equation}
 {\rm Pic}(X)=\bbZ h_1\bigoplus\bbZ h_2\bigoplus_{i=1}^8\bbZ e_i,
\end{equation}
where $e_i=\ep^{-1}(p_i)$, $i=1,\dots,8$, ($p_8=p_\infty$) are exceptional divisors. 
The intersection form $(\,|\,)$ is defined by 
\begin{equation}
 (h_i|h_j)=1-\de_{ij},\quad
 (h_i|e_j)=0,\quad
 (e_i|e_j)=-\de_{ij}.
\end{equation}
The anti-canonical divisor of $X$, denoted by $-K_X$, is uniquely decomposed into the prime divisors:
\begin{equation}
 -K_X=2h_1+2h_2-\sum_{i=1}^8e_i=\sum_{i=0}^4d_i=:\de,
\end{equation}
where
\begin{subequations}\label{eqns:simpleroots_d}
\begin{align}
 &d_0=h_1-e_6-e_8,\quad
 d_1=e_8-e_7,\quad
 d_2=h_2-e_3-e_8,\\
 &d_3=h_1-e_4-e_5,\quad
 d_4=h_2-e_1-e_2.
\end{align}
\end{subequations}
The corresponding Cartan matrix
\begin{equation}
 (d_{ij})_{i,j=0}^4
 =\begin{pmatrix}
 2&-1&0&0&-1\\
 -1&2&-1&0&0\\
 0&-1&2&-1&0\\
 0&0&-1&2&-1\\
 -1&0&0&-1&2
 \end{pmatrix},\quad
 d_{ij}=\cfrac{2(d_i|d_j)}{(d_j|d_j)},
\end{equation}
and Dynkin diagram (see Figure \ref{fig:A4_Dynkin}) are of type $A_4^{(1)}$.
Thus, we can set the root lattice as
\begin{equation}
 Q(A_4^{(1)})=\bigoplus_{i=0}^4\bbZ d_i,
\end{equation}
and identify the surface $X$ as being type $A_4^{(1)}$ in Sakai's classification\cite{SakaiH2001:MR1882403}.
Moreover, we obtain the following root lattice orthogonal to $Q(A_4^{(1)})$:
\begin{equation}\label{eqn:root_symmetry_A4}
 \hat{Q}(A_4^{(1)})=\bigoplus_{i=0}^4\bbZ \al_i,
\end{equation}
where
\begin{subequations}\label{eqns:simpleroots_alpha}
\begin{align}
 &\al_0=h_1+h_2-e_1-e_4-e_7-e_8,\quad
 \al_1=e_1-e_2,\quad
 \al_2=h_1-e_1-e_3,\\
 &\al_3=h_2-e_4-e_6,\quad
 \al_4=e_4-e_5,
\end{align}
\end{subequations}
and 
\begin{equation}
 \de=\al_0+\al_1+\al_2+\al_3+\al_4,
\end{equation}
by searching for elements of Pic$(X)$ that are orthogonal to all divisors $d_i$, $i=0,\dots,4$. 
The root lattice  $\hat{Q}(A_4^{(1)})$ is also of $A_4^{(1)}$-type.

\begin{figure}[t]
\begin{center}
\includegraphics[width=0.32\textwidth]{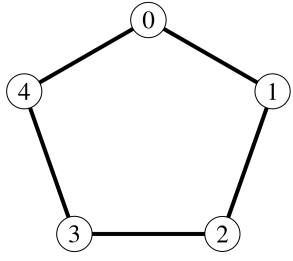}
\end{center}
\caption{Dynkin diagram of type $A_4^{(1)}$.}
\label{fig:A4_Dynkin}
\end{figure}

Let us consider the Cremona isometries for this setting.
A Cremona isometry is defined by an automorphism of Pic$(X)$ which preserves 
\begin{description}
\item[(i)]
the intersection form on Pic$(X)$;
\item[(ii)]
the canonical divisor $K_X$;
\item[(iii)]
effectiveness of each effective divisor of Pic$(X)$.
\end{description}
The reflections $s_i$ for simple roots $\al_i$, $i=0,\dots,4$, defined by the following right actions:
\begin{equation}
 v.s_i=v-\cfrac{2(v|\al_i)}{(\al_i|\al_i)}\,\al_i,
\end{equation}
for all $v\in {\rm Pic}(X)$
and the automorphisms of the Dynkin diagram:
\begin{subequations}
\begin{align}
 &(d_0,d_1,d_2,d_3,d_4;\al_0,\al_1,\al_2,\al_3,\al_4).\sigma
 =(d_2,d_3,d_4,d_0,d_1;\al_4,\al_0,\al_1,\al_2,\al_3),\\
 &(d_0,d_1,d_2,d_3,d_4;\al_0,\al_1,\al_2,\al_3,\al_4).\iota
 =(d_2,d_1,d_0,d_4,d_3;\al_0,\al_4,\al_3,\al_2,\al_1),
\end{align}
\end{subequations}
defined by the following right actions:
{\allowdisplaybreaks
\begin{subequations}
\begin{align}
 &\begin{pmatrix}h_1,h_2,e_1,\dots,e_8\end{pmatrix}.\sigma\notag\\
 &=\begin{pmatrix}h_1,h_2,e_1,\dots,e_8\end{pmatrix}.
 \begin{pmatrix}
  1&1&1&0&0&0&0&0&0&1\\
  1&1&1&0&0&1&0&0&0&0\\
  0&-1&-1&0&0&0&0&0&0&0\\
  0&0&0&0&1&0&0&0&0&0\\
  0&0&0&0&0&0&0&1&0&0\\
  -1&-1&-1&0&0&-1&0&0&0&-1\\
  0&0&0&0&0&0&0&0&1&0\\
  0&0&0&0&0&0&1&0&0&0\\
  0&0&0&1&0&0&0&0&0&0\\
  -1&0&-1&0&0&0&0&0&0&0
 \end{pmatrix},\\
 &\begin{pmatrix}h_1,h_2,e_1,\dots,e_8\end{pmatrix}.\iota\notag\\
 &=\begin{pmatrix}h_1,h_2,e_1,\dots,e_8\end{pmatrix}.
 \begin{pmatrix}
  0&1&0&0&0&0&0&0&0&0\\
  1&0&0&0&0&0&0&0&0&0\\
  0&0&0&0&0&1&0&0&0&0\\
  0&0&0&0&0&0&1&0&0&0\\
  0&0&0&0&0&0&0&1&0&0\\
  0&0&1&0&0&0&0&0&0&0\\
  0&0&0&1&0&0&0&0&0&0\\
  0&0&0&0&1&0&0&0&0&0\\
  0&0&0&0&0&0&0&0&1&0\\
  0&0&0&0&0&0&0&0&0&1
 \end{pmatrix},
\end{align}
\end{subequations}
are Cremona isometries and collectively form extended affine Weyl group of type $A_4^{(1)}$.
Namely, we can easily verify that the following fundamental relations hold:
\begin{subequations}\label{eqns:fundamental_relation_A4}
\begin{align}
 &{s_i}^2=1,\quad (s_i s_{i\pm 1})^3=1,\quad
 (s_i s_j)^2=1,\quad j\neq i\pm 1,\\
 &\sigma^5=1,\quad
 \sigma s_i=s_{i+1} \sigma,\quad
 \iota^2=1,\quad
 \iota s_i=s_{-i} \iota,\quad 
 \sigma \iota=\iota \sigma^{-1},
\end{align}
\end{subequations}
where $i,j\in\bbZ/5\bbZ$.
Note here that the transformations $T_i$, $i=0,\dots,4$, defined by \eqref{eqn:A4_translations} are translations on $\hat{Q}(A_4^{(1)})$:
\begin{equation}
 \al_i.T_i=\al_i-\de,\quad
 \al_{i+1}.T_i=\al_{i+1}+\de,
\end{equation}
where $i\in\bbZ/5\bbZ$.
}

Next, we lift the action of $\widetilde{W}(A_4^{(1)})$ to the birational action.
We first define the variables $f_u$, $f_d$, $g_u$ and $g_d$ by
\begin{equation}
 f=\cfrac{f_u}{f_d},\quad
 g=\cfrac{g_u}{g_d},
\end{equation}
and their polynomial $F_\Lambda$ by
\begin{equation}
 F_\Lambda=F_\Lambda(f_u,f_d,g_u,g_d),
\end{equation}
where $\Lambda=mh_1+nh_2-\sum_{i=1}^8\mu_ie_i$,
which corresponds to a curve of bi-degree $(m,n)$ on $\mathbb{P}^1\times\mathbb{P}^1$
passing through base points $p_i$ with multiplicity $\mu_i$.
For example,
\begin{equation}
 F_{h_1+h_2-e_2-e_5-e_8}=\ga({a_0}^2a_1a_4 f_ug_d+a_2a_3 f_dg_u+a_0a_3a_4f_dg_d),
\end{equation}
where $\ga$ is an arbitrary non-zero complex parameter.
We next define a mapping $\tau$ by the following definition.

\begin{definition}
We define a mapping $\tau$ on the set
\begin{equation}
 M=\set{e_i.w}{w\in\widetilde{W}(A_4^{(1)}),~ i=1,\dots,8}
\end{equation}
by the following:
\begin{description}
\item[(i)]
if under the blowing down map an exceptional line $e_i$ collapses to a base point $p_j$, put 
\begin{equation}
 \tau(e_i)=\tau(e_j);
\end{equation}
\item[(ii)]
if $\Lambda=mh_1+nh_2-\sum_{i=1}^8\mu_ie_i\in\{d_0,d_2,d_3,d_4\}$, then
\begin{equation}
 \cfrac{F_\Lambda(f_u,f_d,g_u,g_d)}{F_\Lambda(1,1,1,1)}=\tau(e_1)^{\mu_1}\cdots\tau(e_8)^{\mu_8},
\end{equation}
which give
\begin{equation}
 f_u=\tau(e_4)\tau(e_5),\quad
 f_d=\tau(e_6)\tau(e_8),\quad
 g_u=\tau(e_1)\tau(e_2),\quad
 g_d=\tau(e_3)\tau(e_8);
\end{equation}
\item[(iii)]
for $\Lambda=mh_1+nh_2-\sum_{i=1}^8\mu_ie_i\in M$, $\tau(\Lambda)$ is defined by
\begin{equation}
 \tau(\Lambda)=\cfrac{F_\Lambda(f_u,f_d,g_u,g_d)}{\tau(e_1)^{\mu_1}\cdots\tau(e_8)^{\mu_8}};
\end{equation}
\item[(iv)]
$w\in\widetilde{W}(A_4^{(1)})$ act on $\tau(\Lambda)$ as
\begin{equation}
 w.\tau(\Lambda)=\tau(\Lambda.w^{-1}),
\end{equation}
where $\Lambda\in M$.
\end{description}
\end{definition}

Finally, Lemma \ref{lemma:tau_A4} follows from the setting
\begin{subequations}
\begin{align}
&\tau_1^{(1)}=\tau(e_2),\quad
\tau_1^{(2)}=\tau(e_3),\quad
\tau_1^{(3)}=\tau(e_6),\quad
\tau_1^{(4)}=\tau(e_5),\quad
\tau_1^{(5)}=\tau(e_7)=\tau(e_8),\\
 &\tau_2^{(1)}=\tau(e_1.\sigma^4)
 =\cfrac{a_0 a_1 (a_3 \tau_1^{(3)} \tau_1^{(5)}+a_0 \tau_1^{(4)} \tau_2^{(3)})}{a_2 {a_3}^2 \tau_2^{(5)}},\\
 &\tau_2^{(2)}=\tau(e_4.\sigma)
 =\cfrac{a_1 a_2 (a_4 \tau_1^{(1)} \tau_1^{(4)}+a_1 \tau_1^{(5)} \tau_2^{(4)})}{a_3 {a_4}^2 \tau_2^{(1)}},\quad
 \tau_2^{(3)}=\tau(e_4),\\
 &\tau_2^{(4)}=\tau(e_1.\sigma)
 =\cfrac{a_3 a_4 (a_1 \tau_1^{(1)} \tau_1^{(3)}+a_3 \tau_1^{(2)} \tau_2^{(1)})}{a_0 {a_1}^2 \tau_2^{(3)}},\quad
 \tau_2^{(5)}=\tau(e_1),
\end{align}
\end{subequations}
and the normalization of the polynomials $F_\Lambda$ to be designed to hold the fundamental relations \eqref{eqns:fundamental_relation_A4}.
We note that the action of $\widetilde{W}(A_4^{(1)})$ on the $\tau$-variables are directly obtained from the definition of the mapping $\tau$.
For example,
\begin{equation}
 s_2(\tau_2^{(5)})
 =\tau(e_1.s_2)
 =\cfrac{\ga'(a_2 a_3 \tau_1^{(3)} \tau_1^{(5)}+a_0 \tau_1^{(4)} \tau_2^{(3)})}{\tau_1^{(2)}},
\end{equation}
where $\ga'$ is an arbitrary non-zero complex parameter.
Moreover, Figure \ref{fig:tau_relation} shows simple relations between the $\tau$-variables.

\begin{figure}[t]
\begin{center}
\includegraphics[width=0.85\textwidth]{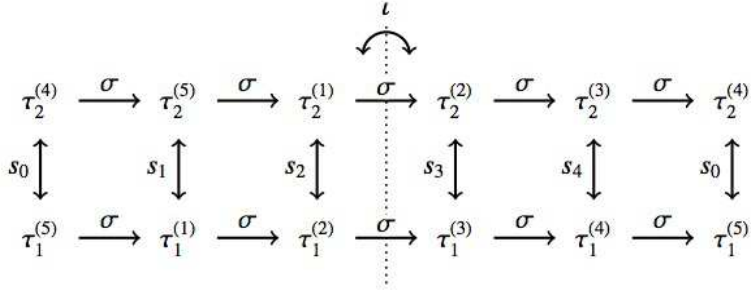}
\end{center}
\caption{Relations between the $\tau$-variables.}
\label{fig:tau_relation}
\end{figure}

\section{The linear action of $\widetilde{W}((A_2\rtimes A_1)^{(1)})$}
\label{section:Q(A2A1)andW(A2)}
In this section, we give explanations of the transformation group $\widetilde{W}((A_2\rtimes A_1)^{(1)})$ 
and its translation part $\langle\rho_1,\rho_2,\rho_3,\rho_4\rangle$ with their linear actions on the root systems.

We here consider the following submodule of the root lattice $\hat{Q}(A_4^{(1)})$ \eqref{eqn:root_symmetry_A4}:
\begin{equation}\label{eqn:root_symmetry_A2A1}
 Q((A_2+A_1)^{(1)})=\bbZ\be_0\bigoplus\bbZ\be_1\bigoplus\bbZ\be_2\bigoplus\bbZ\ga_0\bigoplus\bbZ\ga_1,
\end{equation}
where the simple roots $\be_i$, $i=0,1,2$, and $\ga_i$, $i=0,1$, are defined by
\begin{subequations}
\begin{align}
 &\be_0=\al_0,\quad
 \be_1=\al_1+\al_2,\quad
 \be_2=\al_3+\al_4,\\
 &\ga_0=2\al_1-\al_2+\al_3-2\al_4,\quad
 \ga_1=\al_0-\al_1+2\al_2+3\al_4,
\end{align}
\end{subequations}
and satisfy 
\begin{equation}\label{eqn:de_be_ga}
 \de=\be_0+\be_1+\be_2=\ga_0+\ga_1.
\end{equation}
The root lattices $Q(A_2^{(1)})=\bigoplus_{i=0}^2\bbZ\be_i$ and $Q(A_1^{(1)})=\bigoplus_{i=0}^1\bbZ\ga_i$
are of $A_2^{(1)}$- and $A_1^{(1)}$-types, respectively:
\begin{equation}
 (b_{ij})_{i,j=0}^2=\begin{pmatrix}2&-1&-1\\-1&2&-1\\-1&-1&2\end{pmatrix},\quad
 (c_{ij})_{i,j=0}^1=\begin{pmatrix}2&-2\\-2&2\end{pmatrix},
\end{equation}
where
\begin{equation}
 b_{ij}=\cfrac{2(\be_i|\be_j)}{(\be_j|\be_j)},\quad
 c_{ij}=\cfrac{2(\ga_i|\ga_j)}{(\ga_j|\ga_j)}.
\end{equation}

Let us discuss Cremona transformations for $Q((A_2+A_1)^{(1)})$.
The transformations $w_i$, $i=0,1,2$, $r_i$, $i=0,1$, and $\pi$, defined by \eqref{eqn:WA2_elements},
act on $Q((A_2+A_1)^{(1)})$ as the following:
\begin{subequations}
\begin{align}
 &(\be_0,\be_1,\be_2,\ga_0,\ga_1).w_0=(-\be_0,\be_1+\be_0,\be_2+\be_0,\ga_0,\ga_1),\\
 &(\be_0,\be_1,\be_2,\ga_0,\ga_1).w_1=(\be_0+\be_1,-\be_1,\be_2+\be_1,\ga_0,\ga_1),\\
 &(\be_0,\be_1,\be_2,\ga_0,\ga_1).w_2=(\be_0+\be_2,\be_1+\be_2,-\be_2,\ga_0,\ga_1),\\
 &(\be_0,\be_1,\be_2,\ga_0,\ga_1).r_0=(\be_0,\be_2,\be_1,-\ga_0,\ga_1+2\ga_0),\\
 &(\be_0,\be_1,\be_2,\ga_0,\ga_1).r_1=(\be_1,\be_0,\be_2,\ga_0+2\ga_1,-\ga_1),\\
 &(\be_0,\be_1,\be_2,\ga_0,\ga_1).\pi=(\be_2,\be_1,\be_0,\ga_1,\ga_0).
\end{align}
\end{subequations}
The transformations $w_i$, $i=0,1,2$, 
correspond to the reflections for the simple roots $\be_i$, $i=0,1,2$, respectively,
that is, they satisfy
\begin{equation}
 v.w_i=v-\cfrac{2(v|\be_i)}{(\be_i|\be_i)}\,\be_i,\quad i=0,1,2,
\end{equation}
for all $v\in {\rm Pic}(X)$.
Moreover, the transformation $\pi$ corresponds to the automorphism of the Dynkin diagram:
\begin{equation}
 (d_0,d_1,d_2,d_3,d_4;\be_0,\be_1,\be_2;\ga_0,\ga_1).\pi
 =(d_1,d_0,d_4,d_3,d_2;\be_2,\be_1,\be_0;\ga_1,\ga_0).
\end{equation}
Note that there are no Cremona transformations correspond to the reflections for the simple roots $\ga_i$, $i=0,1$, since
\begin{equation}
 \cfrac{2(h_1|\ga_i)}{(\ga_i|\ga_i)}\,\ga_i=-\cfrac{1}{15}\,\ga_i\, \noin\, {\rm Pic}(X).
\end{equation}
From the fundamental relations \eqref{eqns:fundamental_relation_A4},
we can verified that the group of transformations $\langle w_0,w_1,w_2,r_0,r_1,\pi\rangle$ satisfy the following relations:
\begin{subequations}
\begin{align}
 &{w_i}^2=(w_i w_{i\pm 1})^3=1,\quad
 {r_0}^2={r_1}^2=(r_0r_1)^\infty=1,\quad
 \pi^2=1,\\
 &r_0 w_i=w_{-i}\,r_0,\quad
 r_1  w_i=w_{-i+1}\, r_1,\quad
 \pi w_i=w_{2-i}\,\pi,\quad
 \pi r_0=r_1 \pi,
\end{align}
\end{subequations}
where $i\in\bbZ/3\bbZ$.
We note that the relation $(ww')^\infty=1$ for transformations $w$ and $w'$ means that
there is no positive integer $N$ such that $(ww')^N=1$.
Therefore, transformation group $\langle w_0,w_1,w_2,r_0,r_1,\pi\rangle$ 
forms the extended affine Weyl group of type $(A_2\rtimes A_1)^{(1)}$,
denoted by $\widetilde{W}((A_2\rtimes A_1)^{(1)})$.
Here, $W(A_2^{(1)})=\langle w_0,w_1,w_2\rangle$ and  $W(A_1^{(1)})=\langle r_0,r_1\rangle$ form
affine Weyl groups of types $A_2^{(1)}$ and $A_1^{(1)}$, respectively.
Moreover, $W((A_2\rtimes A_1)^{(1)})=\langle w_0,w_1,w_2,r_0,r_1\rangle$ 
is the semi direct product of $W(A_2^{(1)})$ and $W(A_1^{(1)})$.

The transformations $\rho_i$, $i =1,\dots,4$, defined by \eqref{eqn:def_rho} are translations on $Q((A_2+A_1)^{(1)})$ 
since they act on $Q((A_2+A_1)^{(1)})$ as the following:
\begin{subequations}
\begin{align}
 &(\be_0,\be_1,\be_2,\ga_0,\ga_1).\rho_1=(\be_0,\be_1,\be_2,\ga_0,\ga_1)+(-1,0,1,1,-1)\de,\\
 &(\be_0,\be_1,\be_2,\ga_0,\ga_1).\rho_2=(\be_0,\be_1,\be_2,\ga_0,\ga_1)+(0,1,-1,1,-1)\de,\\
 &(\be_0,\be_1,\be_2,\ga_0,\ga_1).\rho_3=(\be_0,\be_1,\be_2,\ga_0,\ga_1)+(1,-1,0,1,-1)\de,\\
 &(\be_0,\be_1,\be_2,\ga_0,\ga_1).\rho_4=(\be_0,\be_1,\be_2,\ga_0,\ga_1)+(0,0,0,-3,3)\de.
\end{align}
\end{subequations}
Note that 
$\rho_i$, $i =1,\dots,4$, are not translational motions on $\hat{Q}(A_4^{(1)})$ \eqref{eqn:root_symmetry_A4}:
\begin{subequations}
\begin{align}
 &(\al_0,\al_1,\al_2,\al_3,\al_4).\rho_1=(\al_0-\de,\al_1+\al_2+\al_3,-\al_3,-\al_2+\de,\al_2+\al_3+\al_4),\\
 &(\al_0,\al_1,\al_2,\al_3,\al_4).\rho_2=(\al_0,\al_1+\al_2+\al_3,-\al_3+\de,-\al_2,-\al_0-\al_1),\\
 &(\al_0,\al_1,\al_2,\al_3,\al_4).\rho_3=(\al_0+\de,-\al_0-\al_4,-\al_3,-\al_2+\de,-\al_0-\al_1),\\
 &(\al_0,\al_1,\al_2,\al_3,\al_4).\rho_4=(\al_0,-\al_0-\al_4,-\al_3+\de,-\al_2,\al_2+\al_3+\al_4),
\end{align}
\end{subequations}but their squares are translations on $\hat{Q}(A_4^{(1)})$:
\begin{equation}
 {\rho_1}^2=T_0T_2{T_4}^{-1},\quad
 {\rho_2}^2=T_0{T_2}^{-1}T_4,\quad
 {\rho_3}^2={T_0}^{-1}T_2T_4,\quad
 {\rho_4}^2=T_1T_3.
\end{equation}
\def\cprime{$'$} \def\cprime{$'$}

\end{document}